\pdfoutput=1
\documentclass[12pt]{article}
\usepackage[utf8]{inputenc}
\usepackage[dvipsnames,svgnames,x11names]{xcolor}
\usepackage{enumerate}
\usepackage{url} 

\usepackage{changes}
\makeatletter
\@namedef{Changes@AuthorColor}{Black}
\colorlet{Changes@Color}{Black}
\makeatother

\newcommand{\blind}{1}
\usepackage{geometry}
\usepackage{enumitem}
\setlist[enumerate]{itemsep=0mm}
\usepackage[round]{natbib}
\usepackage{ifthen}
\newboolean{tocut}
\setboolean{tocut}{true}% for short version 
\usepackage{authblk}
\usepackage{amsmath,amsfonts,amssymb, amsthm,mathtools}
\allowdisplaybreaks
\usepackage{bbm}
\usepackage{graphicx,subcaption}
\usepackage{float,booktabs}
\usepackage{multirow}
\usepackage{color}
\usepackage[normalem]{ulem}

\usepackage[capitalise,nameinlink]{cleveref}\crefname{appendix}{Supplement}{Supplements}

\usepackage[font=small]{caption}

\newcommand{\alex}[1]{\textbf{\color{Magenta}{Alex: #1}}}
\newcommand{\fbcom}[1]{\textbf{\color{JungleGreen}{Fan: #1}}}

\newtheorem{thm}{Theorem}[section]
\newtheorem{lem}[thm]{Lemma}
\newtheorem{prop}[thm]{Proposition}

\begin{document}

\setlength{\abovedisplayskip}{0pt}
\setlength{\belowdisplayskip}{0pt}
\setlength{\abovedisplayshortskip}{0pt}
\setlength{\belowdisplayshortskip}{0pt}

% more settings from JASA template
\def\spacingset#1{\renewcommand{\baselinestretch}%
{#1}\small\normalsize} \spacingset{1}

%%%%%%%%%%%%%%%%%%%%%%%%%%%%%%%%%%%%%%%%%%%%%%%%%%%%%%%%%%%%%%%%%%%%%%%%%%%%%%

\if1\blind
{
  \title{\bf Likelihood-based Inference for Partially Observed Epidemics on Dynamic Networks}
%   \author{Fan Bu, Allison E. Aiello, Jason Xu\textsuperscript{*}, and  Alexander Volfovsky\textsuperscript{*}
%   %\thanks{The authors gratefully acknowledge \textit{please remember to list all relevant funding sources in the unblinded version}}
%     \hspace{.2cm}%\\
%     %Department of Statistical Science, Duke University%\\
%     %and \\
%     %Author 2 \\
%     %Department of ZZZ, University of WWW
%     }
  \author[1]{Fan Bu}
  \author[2]{Allison E. Aiello}
  \author[1]{Jason Xu\footnote{Joint last authors. $^\S$Corresponding author: \texttt{jason.q.xu@duke.edu}}$^\S$}
  \newcommand\CoAuthorMark{\footnotemark[\arabic{footnote}]}
  \author[1]{Alexander Volfovsky\protect\CoAuthorMark}
  \affil[1]{Department of Statistical Science, Duke University}
   \affil[2]{Gillings School of Global Public Health, University of North Carolina at Chapel Hill\thanks{The eXFLU data were supported by U01 CK000185, AV and AA were partially supported by R01 EB025021, AV and FB were partially supported by W911NF1810233, JX was partially supported by DMS 1606177.}}
  \date{}
  \maketitle
} \fi

\if0\blind
{
  \bigskip
  \bigskip
  \bigskip
  \begin{center}
    {\LARGE\bf Likelihood-based Inference for Partially Observed Epidemics on Dynamic Networks}
\end{center}
  \medskip
} \fi

\bigskip

%\title{Modeling Dynamic Network Epidemic Processes with Incomplete Observations}
%\title{Likelihood-based Inference for Partially Observed SIR Processes on Adaptive Contact Networks}
%\author{Fan Bu, Alexander Volfovsky, and Jason Xu}
% Author names arranged by alphabetical order
%\affil{\emph{Department of Statistical Science, Duke University}}
%\date{}

%\maketitle

\begin{abstract}
    
We propose a generative model and an inference scheme for epidemic processes on dynamic, adaptive contact networks. Network evolution is formulated as a link-Markovian process, which is then coupled to an individual-level stochastic SIR model, in order to describe the interplay between the dynamics of the disease spread and the contact network underlying the epidemic. A Markov chain Monte Carlo framework is developed for likelihood-based inference from partial epidemic observations, with a novel data augmentation algorithm specifically designed to deal with missing individual recovery times under the dynamic network setting. Through a series of simulation experiments, we demonstrate the validity and flexibility of the model as well as the efficacy and efficiency of the data augmentation inference scheme. The model is also applied to a recent real-world dataset on influenza-like-illness transmission with high-resolution social contact tracking records.

\end{abstract}

% more stuff from JASA template
\noindent%
\color{black}
\vfill

\newpage
\spacingset{1.5} % DON'T change the spacing!

\section{Introduction}\label{sec:intro}

The vast majority of epidemiological models, such as the well-known susceptible-infectious-recovered (SIR) model, rely on compartmentalizing individuals according to their disease status %meaning that individual level information is lost 
\citep{kermack1927contribution}. Classically, such models describe population-level behavior under a ``random mixing'' assumption that an infectious individual can spread the disease homogeneously to any susceptible individual \citep{kermack1927contribution, bailey1975mathematical, anderson1992infectious}. In the last two decades, an alternative assumption---that the disease is transmitted through links in a contact network---has gradually gained popularity. \added{It has been found that the contact network structure can fundamentally impact the behavior of epidemic processes} \citep{wallinga1999perspective,edmunds1997mixes,edmunds2006mixing,mossong2008social,volz2008epidemic,melegaro2011types}; \added{on the other hand, the network structure can in turn be influenced by disease status of individuals as well} \citep{bell2006world,funk2010modelling,eames2010impact,van2013impact}.% \added{, thus leading to increased consideration and analysis of networks in transmittable disease research.}

\added{
This growing interest---alongside technological advances in mobile data---has spurred efforts on collecting high-resolution data that inform the dynamics of the contact network \citep{vanhems2013estimating,Barrat2014Measuring,Voirin2015Combining,kiti2016quantifying,aiello2016design,ozella2018close}. Data of this type has most recently been collected by \citet{israelData1} and \citet{KoreaData}.
However, there is a gap between the demand to analyze such emergent data and available methods: 
% There are many open problems in inference for epidemic models on networks to be addressed; 
a recent review by \citet{BrittonInference20} outlines possible approaches 
to inference for epidemic models on networks
and calls to attention considerable challenges---in particular, key terms such as transition probabilities that appear in central quantities such as likelihood expressions are unavailable. 
Accounting for the relationship between disease spread and the underlying contact network during inference 
% not only 
is crucial to accurately estimating parameters describing the inherent properties of the disease, 
% but 
and moreover has direct practical implications on epidemic control and intervention. Policies such as quarantine or suppression are naturally described as changes to the contact network, and yet modeling approaches are largely restricted to prospective simulations and/or analysis based on static networks in lieu of direct inference from modern data. Recent examples can be seen in analyses for COVID-19 \citep{ImperialCollege2020} and for MERS-CoV transmissions \citep{YangMERS20}. }

\added{To address this methodological gap, we develop a statistical model that describes the mutual interplay between SIR-type epidemics and an underlying dynamic network, together with a tractable inferential framework to fit such models to modern time-resolved datasets. In particular, we propose a stochastic generative model that can be fit to data using likelihood-based methods. We present a Bayesian data augmentation scheme that accommodates partial observations such as missing recovery times that are common in real data while quantifying uncertainty in estimated parameters.
}

\added{The majority of existing work on epidemic processes over networks adopt a deterministic approach based on ordinary differential equations (ODEs) models  %that do \textbf{not} enable likelihood-based inference, nor provide any measure of uncertainty (examples include 
\citep{kiss2012modelling,ogura2017optimal,van2010adaptive,tunc2013epidemics,group2006nonpharmaceutical,volz2007susceptible,shaw2008fluctuating,volz2008sir}. Such approaches do not provide a measure of uncertainty, and do not offer probabilistic interpretations. Our framework builds upon previous stochastic (and likelihood-based) methods that either do not consider network dynamics \citep{britton2002bayesian,dong2012graph,fan2015hierarchical,fan2016unifying}, or do not model the contact network at all \citep{cauchemez2006s,hoti2009outbreaks,britton2010stochastic,fintzi2017efficient,ho2018birth}. Moreover, the proposed inference scheme accommodates partially observed data, with a focus on unknown recovery times in this paper. Handling missing data even \textbf{without} network constraints is already challenging, and often requires simplifying assumptions \citep{cauchemez2008likelihood} or computationally intensive simulation-based inference \citep{he2010plug,andrieu2010particle, pooley2015using}. }

%%%%%%%%%%%%%%%%%%%%%%%%%%%%%%
% commented out 03/26/2020
%\deleted{We choose stochastic models over deterministic models widely used in existing literature to account for the randomness in small-scale epidemic outbreaks and allow for measures of uncertainty in estimation}\citep{britton2010stochastic}.

%Networks can flexibly and informatively represent human interactions, contact patterns, and social structures, which can fundamentally impact the behavior of epidemic processes \citep{wallinga1999perspective,edmunds1997mixes,edmunds2006mixing,mossong2008social,volz2008epidemic,melegaro2011types}. Also, there is evidence that the population network structure can be influenced by infection events as well \citep{bell2006world,funk2010modelling,eames2010impact,van2013impact}. Numerous recent epidemiological methods consider both disease transmission and the contact network as dynamic processes \citep{bansal2010dynamic,pastor2015epidemic,masuda2013predicting,enright2018epidemics, masuda2017temporal,dong2012graph,fan2015hierarchical}, and even as \textbf{coupled} processes: disease spread is impacted by the evolving network structure, which also \textbf{adapts} to contagion progression in that individuals with different health statuses adopt different social behaviors \citep{kiss2012modelling,ogura2017optimal,van2010adaptive,tunc2013epidemics,group2006nonpharmaceutical,volz2007susceptible}.

\ifthenelse{\boolean{tocut}}{
%\alex{Alex: a bunch has been cut.} %\fbcom{Fan: a bit more background is added to the part above, and the paragraph below is slightly edited to make the narrative a bit smoother.}
}{
Static networks where interactions between individuals are assumed fixed and permanent were first introduced to epidemiological models to account for heterogeneous contact patterns and stratified population structures.
% and have since contributed to valuable results that would generalize or deviate from previous ones.
Theoretical works have shown how network structures influence important epidemiological features such as the epidemic threshold, final disease size and equilibrium states \citep{eames2002modeling,PhysRevE.66.047104,PhysRevLett.86.3200,van2009virus,simon2011exact,PhysRevE.86.016116}. Simulation-based studies have investigated potential outcomes of specific transmittable diseases and possible intervention strategies on various networked populations \citep{halloran2002containing,eubank2004modelling, meyers2005network, auranen2004modelling, read2003disease, eames2004monogamous, dube2009review, martinez2009social}. 

%, where interactions between individuals are assumed fixed and permanent, were first introduced to epidemiological models. The consideration of heterogeneous contact patterns and stratified population structures in the analysis of contagion processes has produced numeric theoretical results
%on summary statistics that are of particular importance in classical models, such as the basic reproductive number, the epidemic threshold, the final disease size, and equilibrium states 
%\citep{eames2002modeling,PhysRevE.66.047104,PhysRevLett.86.3200,van2009virus,simon2011exact,PhysRevE.86.016116}, 
%and census-based contact networks or idealized random network models have been used in epidemic simulation, prediction, or intervention
%\citep{halloran2002containing,eubank2004modelling, meyers2005network, auranen2004modelling, read2003disease, eames2004monogamous, dube2009review, martinez2009social}. 

However, static networks
% , where interactions between individuals are assumed fixed and permanent, 
are restrictive and unrealistic. 
Contact networks are dynamic in nature, and thus time-varying interactions between members of a host population cannot be captured by a static framework. This has motivated a considerable amount of works to model the spread of infectious diseases on temporal networks instead of static ones  \citep{bansal2010dynamic,pastor2015epidemic,masuda2013predicting,enright2018epidemics, masuda2017temporal}.
%However, static networks are in fact restrictive and lacking due to the incapability of capturing the time-varying nature of interactions between members of a host population. Temporal networks, in comparison, provide a framework that explicitly describes and even models the changes in contacts over time. A considerable amount of recent works have investigated the effects of temporal networks on the spread of infectious diseases \citep{masuda2017temporal}, and there have been multiple reviews summarizing results and progresses \citep{bansal2010dynamic,pastor2015epidemic,masuda2013predicting,enright2018epidemics}. 
In addition, there has been an awareness that network structures can be influenced by infection events \citep{bell2006world,funk2010modelling,eames2010impact,van2013impact}, thus giving rise to some of the most recent methods that regard networks and epidemics as \textbf{coupled} processes: social interactions may change as a result of contagion progression (for example, illnesses are associated with reduced social contacts, as observed by \cite{van2013impact}), and disease spread is further impacted by the evolving network structure; the contact network in this scenario is not merely dynamic, but \textbf{adaptive}, where links can be established or dissolved based on each individual's health status. %Therefore, studying the interplay of disease dynamics on networks and the dynamics of networks not only informs epidemiologists on epidemic prediction and intervention, but also potentially helps social scientists understand the impact of contagion-like diffusion processes on social behaviors.
Studying the interplay of disease dynamics \textbf{on} networks and the dynamics \textbf{of} networks has produced theoretical guidelines on disease control and social interventions \citep{kiss2012modelling,ogura2017optimal,van2010adaptive,tunc2013epidemics,group2006nonpharmaceutical} and offered great insights into the impact of contagion-like diffusion processes on social networks \citep{Kiss2017}.

Despite the increasing interest in the coupled processes of epidemics and adaptive networks, almost all previous works have focused on theoretically deriving or numerically simulating the forward behavior of the system, rather than conducting inference from real-world observations. This is
% , in practice, 
largely due to the long-standing difficulty of accessing high-quality data on dynamic social contacts. However, thanks to recent technological advances in mobile devices, tracking social interactions in the context of epidemiological research has become possible and even convenient. Modeling disease spread on dynamic and adaptive networks is therefore no longer restricted to prospective extrapolation with presumed parameters and conditions, but is moving toward retrospective inference %and predictive modeling 
based on real data. In the last few years, multiple observational studies have been conducted to estimate transmission pathways or evaluate non-pharmaceutical intervention strategies by collecting high-resolution data on social contacts \citep{vanhems2013estimating,aiello2016design} (\fbcom{Need a few more references on similar studies. - Fan}). Such data provide rich and valuable information on the complex mechanisms of infection progression and network evolution, but they also present several major challenges that suitable methodologies have to address: first, network changes should be modeled along with epidemic events, since they directly determine the contact patterns and are potentially influenced by disease transmission; second, similar observations are often made on a small-scale community (of dozens or hundreds in size, rather than thousands) with relatively few disease cases, so a certain level of randomness is to be expected; third, there is almost always some missingness in the data, particularly in terms of the exact times of individual disease episodes. 

\alex{proposal: move the next two paragraphs to be the second and third paragraph of the paper. then take everything above and make it into the literature review section. it is well written but you need to get to the contribution before page 3. if need to save space, cut everything above (except references for the type of data we will address}
}

The rest of the paper is organized as follows: Section~\ref{sec:background} reviews background on epidemic models (focusing on the stochastic SIR model) and dynamic network processes. Section~\ref{sec:model} formulates the generative model and derives maximum likelihood estimators as well as Bayesian posterior distributions based on the complete data likelihood. Section~\ref{sec:missrecov} describes a Bayesian inference scheme that deals with incomplete observations on individual recovery times. Section~\ref{sec:sim-experiments} and \ref{sec:real-experiments} present experiment results on simulated datasets and a real-world dataset. Finally Section~\ref{sec:discussion} provides further discussions.

\section{Background}
\label{sec:background}
%\subsection{Compartmental Epidemiological Models}
\paragraph{Compartmental Epidemiological Models}

 %compartmentalizing members of a population by their disease status \citep{kermack1927contribution, bailey1975mathematical, anderson1992infectious}.
 %The vast majority of epidemiological models are based on
%compartmentalizing individuals into non-over lapping subsets according to their disease statuses.
Compartmental models divide individuals into non-overlapping subsets according to their disease statuses.
In classical models, the change in these subpopulations over time are described by ordinary differential equations (ODEs) \citep{hethcote2000mathematics}. One widely used model is the susceptible-infectious-recovered (SIR) model, which assumes three disease statuses---susceptible ($S$), \added{infectious} ($I$), and recovered (or removed, $R$). On a closed population of $N$ individuals (with $N$ sufficiently large), the dynamics of the deterministic SIR model can be expressed as
\begin{equation}
\label{eq:SIR-ODE}
    \frac{dS(t)}{dt} = -\beta S(t)I(t), \quad \frac{dI(t)}{dt} = \beta S(t)I(t) - \gamma I(t),
\end{equation}
where $S(t)$ and $I(t)$ refer to the number of susceptible and infectious individuals at time $t$, respectively, and the number of recovered individuals satisfies $R(t) = N -\left[S(t)+I(t)\right]$. \added{The parameters describe the epidemic mechanistically: here $\beta$ is interpreted as the rate of disease transmission \textbf{per contact between} an $S$ individual and an $I$ individual, and $\gamma$ is the rate of recovery for an $I$ individual.}

%The SIR model is most commonly applied to epidemics of diseases that induce long-term immunity after attack, such as mumps and measles (some citation should go here). However, it is also appropriate for infectious diseases that confer a certain level of immunity which spans across the time scope of interest. For example, an individual who contracts a particular strain of influenza develops immunity to it (some other citation), so the SIR model can be used to describe the spread of a type of influenza during a flu season.

By setting the growth rate of infection to be proportional to $S(t)I(t)$, the model in (\ref{eq:SIR-ODE}) implicitly assumes that %the population is well-mixed, such that 
any two members can interact with each other. 
This assumption is easily violated in reality, where an individual only maintains contact with a limited number of others. Moreover, the differential equations can only account for the average, expected behavior of the process, but the transmission of an infectious disease exhibits randomness and uncertainty by nature.

To account for the underlying network structure of a population as well as the random nature of an epidemic process, we adopt an \textbf{individual-level}, \textbf{stochastic} variation of the SIR model, similar to that used in \cite{auranen2000transmission}. An individual of status $S$ (susceptible) at time $t$ ($>0$) changes disease status to $I$ (infectious) at time $t+h$ ($h$ is an infinitesimal quantity) with a probability that is dependent on both the infection rate $\beta$ \textbf{and} his/her contacts at time $t$. An infectious individual at time $t$ becomes a member of the $R$ (recovered) sub-population at time $t+h$ with a probability determined by the recovery rate $\gamma$. Specifically, for any susceptible individual $p_1$ and infectious individual $p_2$ in the population at $t$, conditioned on the current overall state of the process, $\mathcal{Z}_t$, then
\begin{equation}
    \label{eq:infection-background}
        Pr(p_1 \text{ gets infected by } p_2 \text{ at } t+h \mid \mathcal{Z}_t) = \beta h + o(h)
\end{equation}
if $p_1$ and $p_2$ are in contact at $t$, and
\begin{equation}
\label{eq:recovery-background}
        Pr(p_2 \text{ recovers at } t+h \mid  \mathcal{Z}_t) = \gamma h + o(h).
\end{equation}

%\subsection{Basic Network Concepts}
\paragraph{Basic Network Concepts}

A network, or a graph, is a two-component set, $\mathcal{G} = \{\mathcal{V},\mathcal{E}\}$, where $\mathcal{V}$ is the set of $N$ nodes and $\mathcal{E}$ is the set of links. A network can be represented by its ``adjacency matrix'', $\mathbf{A}$, where $\mathbf{A}_{ij}=1$ indicates there is a link from node $i$ to $j$. Since most infectious diseases can be transmitted in both directions through a contact, we assume that the adjacency matrix is \textbf{symmetric}, $\mathbf{A}_{ij}=\mathbf{A}_{ji}$.

A special network structure is the fully connected network (or the complete graph), $\mathcal{K}_N$, and its adjacency matrix $\mathbf{A}$ satisfies $\mathbf{A}_{ij} = 1$ for any $i \neq j$. This network structure corresponds to the widely adopted ``random mixing'' assumption in epidemiological models, %which states that any susceptible individual is exposed to \textbf{all} the infectious ones. Such an assumption, in fact, is unrealistic and restrictive, and therefore 
which, as stated before, may be restrictive and unrealistic. Therefore,
in the rest of the paper, we instead consider \textbf{arbitrary} network structures underlying the population. 

\paragraph{Temporal and Adaptive Networks}
%\subsection{Temporal and Adaptive Networks}

Interactions between individuals are dynamic in nature, and such dynamics is important when modeling epidemic processes (\cite{masuda2017temporal}; also as demonstrated later in Section~\ref{sec:infer-complete}).
We consider a continuous-time link-Markovian model for temporal networks \citep{clementi2010flooding,ogura2016stability}. \added{Following the symmetric network assumption stated above,} for two individuals $i$ and $j$ \added{($i < j$)} who are not in contact at time $t$, they form a link at time $t+h$ ($h \ll 1$) with probability $\alpha h$, where $\alpha$ is the link activation rate. Similarly, if there is an edge between $i$ and $j$ at time $t$, then the edge is deleted at time $t+h$ with probability $\omega h$, where $\omega$ is some link termination rate.

If, instead, individuals establish and dissolve their social links with rates that vary according to their disease statuses, then the evolution of the network is coupled to the epidemic process and thus becomes \textbf{adaptive}.
This mechanism can be described via instantaneous rates of single-link changes. For any two individuals $i$ and $j$, their corresponding entry in the adjacency matrix is modeled as a $\{0,1\}$-valued Markov process, $\mathbf{A}_{ij}(t), t>0$. Suppose that at time $t$, $i$ is of status $A$, $j$ is of status $B$,\footnote{Here $A,B \in \{S,I,R\}$\added{, and we only consider $i<j$ since the network is assumed symmetric}.} then for an infinitesimal quantity $h$,

\begin{align}
    \label{eq:activation-background}
    Pr(\mathbf{A}_{ij}(t+h)\added{=\mathbf{A}_{ji}(t+h)}=1 \mid \mathbf{A}_{ij}(t)=0) &= \alpha_{AB} h + o(h);\\
    \label{eq:termination-background}
    Pr(\mathbf{A}_{ij}(t+h)\added{=\mathbf{A}_{ji}(t+h)}=0 \mid \mathbf{A}_{ij}(t)=1) &= \omega_{AB} h + o(h).
\end{align}

\added{Here $\alpha_{AB}(= \alpha_{BA})$ is the activation rate for an $A$-$B$ type link, and similarly, $\omega_{AB}(= \omega_{BA})$ is the termination rate for an $A$-$B$ type link.}

%\section{Adaptive Network Epidemic Processes and Inference with Complete Data}
\section{Epidemic Processes over Adaptive Networks}
\label{sec:model}
\subsection{The Generative Model}
\label{sec:model-overview}
In this subsection, we lay out a \added{stochastic data generative process (referred to as the ``generative model'')} for the joint evolution of an individualized SIR process on a networked population and the dynamics of the contact network. 
\added{In contrast to the ODE literature and existing network models described in Section \ref{sec:intro},
the key feature of the model is the \textbf{interplay} between epidemic progression and network adaptation}. On one hand, transmission of infection depends on the existence of susceptible-infectious links, which may change through time; on the other hand, network links temporally update in a manner that in turn depends on individual disease status. 

We formulate this complex process as a continuous-time Markov chain that comprises all individual-level events described in Section \ref{sec:background}. The joint evolution of the individual Poisson processes described by (\ref{eq:infection-background})-(\ref{eq:termination-background}) can be described via a competing risks construction. 
%The model can be described via the competing risks construction of competing Poisson processes with exponentially-distributed wait times at the individual level, as described in (\ref{eq:infection-background})-(\ref{eq:termination-background}). 
By the Markov property, the time until each type of event has an exponential waiting time, and thus the time to next event in the joint process remains exponentially distributed \citep{guttorp2018stochastic}. Events occur stochastically and are of one of the four types:
%Throughout the process, the system is updated stochastically with one event (for one individual or a pair) at a time, and each event is of one of the four types: 

\vspace*{-0.1in}
\begin{itemize}
\setlength\itemsep{0em}
    \item \textbf{Infection}: The disease is transmitted through a link between an $S$ (susceptible) and an $I$ (infectious) individual ($S$-$I$ link) with rate $\beta$;
    \item \textbf{Recovery}: Each $I$ individual recovers with rate $\gamma$ independently;
    \item \textbf{Link activation}: A link is formed at rate $\alpha_{AB}$\added{($=\alpha_{BA}$)} between an individual of status $A$ and another of status $B$ who are not connected, where $A,B \in \{S,I,R\}$;
    %\footnote{Under the SIS model, $A,B \in \{S,I\}$.}
    \item \textbf{Link termination}: An existing link is removed at rate $\omega_{AB}$\added{($=\omega_{BA}$)} between an individual of status $A$ and another of status $B$, where $A,B \in \{S,I,R\}$.
\end{itemize}
This formulation will allow for joint inference of both disease spread and network evolution. As illustrated in the next subsection, inference is straightforward when all the events are fully observed. Furthermore, this formulation implies a relatively simple generative process at the population level. 
Conditioned on the current state of the process $\mathcal{Z}_t$ at time $t\,(>0)$, the very next event of the entire process is the 
\textbf{earliest} event that occurs among the four competing processes by the superposition property:
\vspace*{-0.1in}
\begin{itemize}
\setlength\itemsep{0em}
    \item \textbf{Infection}: An infection occurs with rate $\beta S\kern-0.14em I(t)$, where $S\kern-0.14em I(t)$ is the number of $S$-$I$ links at time $t$;
    \item \textbf{Recovery}: A recovery occurs with rate $\gamma I(t)$, where $I(t)$ is the number of infectious individuals at time $t$;
    \item \textbf{Link activation}: An $A$-$B$ link is established with rate $\alpha_{AB}M_{AB}^d(t)$, where $M_{AB}^d(t)$ is the number of disconnected $A$-$B$ pairs at time $t$;
    \item \textbf{Link termination}: An $A$-$B$ link is dissolved with rate $\omega_{AB}M_{AB}(t)$, where $M_{AB}(t)$ is the number of connected $A$-$B$ pairs at time $t$.
\end{itemize}
We may interpret this generative model as a generalization of two simpler models. If we set $\alpha_{AB} \equiv \alpha$ and $\omega_{AB} \equiv \omega$ for any status $A$ and $B$, the coupled process reduces to a \textbf{decoupled} process, where network evolution is \textbf{independent} of individual disease status. Moreover, if we fix $\alpha \equiv \omega \equiv 0$, the process is further reduces to an SIR process over a \textbf{static} network.

Here we assume that the population size $N$ is fixed, and \added{that at $t=0$, the initial network $\mathcal{G}_0$ as well as $I(0)$ initial infection cases are observed}. We summarize a list of  model parameters and notation %used for model formulation and inference 
in Table~\ref{tab:notation}.

\begin{table}[H]
    \centering
    \small
    \caption{Table of parameters and notation.}
    \begin{tabular}{ll}
    \toprule
        \textbf{Parameter} &  \textbf{Description}\\
    %\midrule
    \hline
        $\beta$ &  infection rate \\
        $\gamma$ &  recovery rate\\
        $\alpha$ & link activation rate for a currently disconnected pair \\
        $\omega$ & link termination rate for a currently connected pair \\
        $\alpha_{AB}$ & link activation rate for a currently disconnected A-B pair\\
        $\omega_{AB}$ & link termination rate for a currently connected A-B pair\\
    \toprule
        \textbf{Notation} &  \textbf{Description}\\
    %\midrule
    \hline
        $N$ & total population size (assumed to remain fixed throughout the process) \\
        $T_\text{max}$ & maximum observation time \\
        $\mathcal{Z}_t$ & state of the process at time $t$ (including the epidemic status\\
        &  of every individual and the social network structure at time $t$) \\
        $\mathcal{G}_t$ & social network structure (a graph) at time $t$\\
        $S(t), I(t)$ & numbers of susceptible/infectious individuals in the population at time $t$\\
        $H(t)$ & number of healthy (not infectious) individuals in the population at time $t$ \\
        $I_k(t)$ & number of infectious individuals in person $k$'s neighborhood at time $t$\\
        $S\kern-0.14em I(t)$ & number of $S$-$I$ links  in the network at time $t$\\
        $M(t)$ & total number of edges in the network at time $t$\\
        $M_{AB}(t)$ & number of $A$-$B$ links at time $t$\\
        $M_{AB}^d(t)$ & number of disconnected $A$-$B$ pairs at time $t$ \\
        $n_E, n_R$ & counts of infection events and recovery events in the process\\
        $n_N$ & count of network events in the process\\
        & (each event is the activation or termination of a single link)\\
        $C, D$ & total counts of link activation/termination \\
        $C_{AB}, D_{AB}$ & counts of link activation/termination events for $A$-$B$ pairs \\
    \bottomrule
    \end{tabular}
    \label{tab:notation}
\end{table}

\subsection{Complete Data Likelihood and Parameter Inference}
\label{sec:likelihood-and-inference}
\paragraph{Derivation of complete data likelihood}
\added{Suppose $n_E$ infection events and $n_R$ recovery events are observed in total.
Let $i_k$ be the infection time for individual $k$ ($k = 1,2,\ldots, n_E$), $r_k$ be $k$'s recovery time (if $r_k > T_{\max}$, $k$'s recovery is not observed), and without loss of generality, set $i_1 = 0$. }
Recall that the widely used ``random mixing'' assumption %\footnote{The ``random mixing'' assumption suggests that an infectious person can spread the disease to \textbf{any} other susceptible individual in the population.} 
in classical epidemiological models is equivalent to assuming that the contact network is a complete graph, $\mathcal{K}_N$, and the \added{individual-based} complete data likelihood under this assumption is \footnote{\added{Note that this expression differs from the \textbf{population-level} complete likelihood \citep{becker1999statistical}. When epidemic events are tied to individuals, i.e. a recovery time is associated to a specific infection event, we must use the individual-based likelihood instead (similar to that in \cite{auranen2000transmission}).}}
\begin{equation*}
    \mathcal{L}(\beta, \gamma) = p(\text{epidemic events}|\beta, \gamma) \nonumber\\
    = \gamma^{n_R} \prod_{k=2}^{n_E}\left[\beta I(i_k)\right] \exp\left(-\int_0^{T_{\max}}\left[\beta S(u)I(u) + \gamma I(u)\right]du\right).
\end{equation*}

% \jxadd{\bf I restructured this to condition on the network once rather than walking through the static and dynamic cases. Let me know if this is ok:}
To account for the contact network, let $\mathcal{G}_t$ be an arbitrary network, and begin by assuming that the \textbf{entire network process} $\{\mathcal{G}_t: 0<t<T_{\max}\}$ is fully observed. % (or all the network changes during the contagion process are fully observed),
Explicitly accounting for the number of infectious contacts per individual at the time of infection as well as the total number of $S$-$I$ links in the system, the complete data likelihood becomes:

%Now assume that the network is an arbitrary \textbf{static} network $\mathcal{G}$. 
\begin{equation}
\label{eq:epi-alone-lik}
    \mathcal{L}(\beta, \gamma|\mathcal{G}) = \gamma^{n_R} \prod_{k=2}^{n_E}\left[\beta I_k(i_k)\right] \exp\left(-\int_0^{T_{\max}}\left[\beta S\kern-0.14em I(u) + \gamma I(u)\right]du\right).
\end{equation}
Here $I_k(i_k)$ denotes the number of \textbf{infectious neighbors of} person $k$ at his time of infection $i_k$ and $S\kern-0.14em I(t)$ denotes the number of $S$-$I$ links in the system at time $t$. 
We see that the dynamic nature of the network is implicitly subsumed into the terms $I_k(i_k)$'s and $SI(u)$. To clarify this point, note that the same likelihood holds for a \textit{static} network $\mathcal{G}$. As neighborhoods $I_k(i_k)$ are fixed in the static case, one could further simplify \eqref{eq:epi-alone-lik} using a constant $I_k(i_k) = I_k$ for all times $i_k$.

% This is resolved on 03/31/2020
%\alex{I wonder if we want to write $SI(t)$ as $S\kern-0.14em I(t)$---note the subtle difference in the distance between the $S$ and the $I$.}

%If the network is not static, but the \textbf{entire network process} $\{\mathcal{G}_t: 0<t<T_{\max}\}$ were given (or all the network changes during the contagion process are fully observed), then the form of the  likelihood \eqref{eq:epi-alone-lik} remains unchanged \textbf{conditional} on $\{\mathcal{G}_t\}$. That is, the dynamic nature of the network is implicitly subsumed into the terms $I_k(i_k)$'s and $SI(u)$. Note that the neighborhoods $I_k(i_k)$ are fixed under a static network, i.e. one could simplify the expression using a constant $I_k(i_k) = I_k$ for all times $i_k$.

Equation \eqref{eq:epi-alone-lik} serves as a point of departure toward network dynamics. As a stepping stone, we first consider the simpler \textbf{decoupled} case in which the network and epidemic evolve \textbf{independently}. Here the edge activation rate $\alpha$ and deletion rate $\omega$\added{, as well as total number of activated and terminated edges denoted $C$ and $D$, do not depend on disease status}. Given an initial network $\mathcal{G}_0$, the network process likelihood can be easily written as
\begin{align}
\label{eq:decoupled-net-lik}
    &\mathcal{L}(\alpha, \omega|\mathcal{G}_0) = p(\text{network events}|\alpha, \omega,\mathcal{G}_0) \nonumber\\
    =& \alpha^C \omega^{D} \prod_{\ell=1}^{n_N}\left[\left(\frac{N(N-1)}{2}-M(s_{\ell})\right)^{1-D_\ell}M(s_{\ell})^{D_\ell}\right] \nonumber \\
    &\times \exp\left(-\alpha\frac{N(N-1)}{2}T_{\max}+ (\alpha-\omega)\int_{0}^{T_{\max}}M(u)du\right).
\end{align}
Here $s_{\ell}$ is the time of the $\ell$th network event, and $D_{\ell} = 1$ if this event is a \textbf{link termination} and otherwise $D_{\ell} = 0$.
%Therefore, when the epidemic process and network process are decoupled, 
\added{By independence, the complete data likelihood in this decoupled case is simply a product of Equations (\ref{eq:epi-alone-lik}) and (\ref{eq:decoupled-net-lik})}:
\begin{align}
\label{eq:decoupled-process-lik}
    &\mathcal{L}(\beta, \gamma, \alpha, \omega|\mathcal{G}_0) %= p(\text{epidemic events},\text{network events}|\beta, \gamma, \alpha, \omega,\mathcal{G}_0) \nonumber \\
    = p(\text{epidemic events}|\beta, \gamma,\mathcal{G}_t) \cdot  p(\text{network events}|\alpha, \omega,\mathcal{G}_0) \nonumber \\
    =& \beta^{n_E - 1} \gamma^{n_R} \alpha^C \omega^{D} \prod_{k=2}^{n_E}\left[ I_k(i_k)\right] \prod_{\ell=1}^{n_N}\left[\left(\frac{N(N-1)}{2}-M(s_{\ell})\right)^{1-D_\ell}M(s_{\ell})^{D_\ell}\right]  \\
    & \times \exp\left(-\int_0^{T_{\max}}\left[\beta SI(u) + \gamma I(u) + (\omega - \alpha)M(u)\right]du -\alpha\frac{N(N-1)}{2}T_{\max}\right) \nonumber.
\end{align}
\added{Finally, we allow
% we relax these assumptions to consider the coupled process---
link activation and termination 
% are now
to be
dependent on individual disease status, yielding an \textbf{adaptive} network.} We introduce some notation;
it is natural to assume that the $S$ and $R$ populations behave identically from the perspective of the network process:
% \begin{equation*}
%     \alpha_{R \cdot } \equiv \alpha_{S \cdot},  \text{ and } \omega_{R \cdot } \equiv \omega_{S \cdot}.
% \end{equation*}
\begin{equation*}
\added{
    \alpha_{R A } \equiv \alpha_{S A},  \text{ and } \omega_{R A } \equiv \omega_{S A},\,  \forall A\in \{S,I,R\}.
}
\end{equation*}
%Referring to them as the healthy population ; 
Let $H(t) = R(t)+S(t) = N - I(t)$ denote the number of such ``healthy" individuals at time $t$. Naturally the term ``$H$-$H$ link'' represent an $S$-$S$ link, an $S$-$R$ link, or an $R$-$R$ link, and the term ``$H$-$I$ link'' represents an $S$-$I$ link or $R$-$I$ link. We also define $g(p,t)$ as the indicator function of infectiousness, i.e. $g(p,t) = 1$ if person $p$ is infected at time $t$ and $g(p,t) = 0$ otherwise.

\added{Denote the ordered epidemic and network events together  as $\{e_j = (t_j, p_{j1},p_{j2})\}_{j=1}^n$, with $n = n_E + n_R + n_N$. Here $t_j$ ($j=1,2,\ldots,n$) denote the event times %with $t_1 < t_2 < \ldots < t_n$ 
and $t_1=0$ is the infection time of the first patient.} If $e_j$ is a network event,  $p_{j1}$ and $p_{j2}$ are the two individuals getting connected or disconnected, and if $e_j$ is an epidemic event, let $p_{j1}$ be the person getting infected or recovered and set $p_{j2}=0$. Furthermore let event type indicators $F_j, C_j, D_j$ take the value $1$ only if $e_j$ is an infection, a link activation, and a link deletion, respectively, and $0$ otherwise.

The contribution of all network events to the complete data likelihood is in essence of the same form as (\ref{eq:decoupled-net-lik}), except that for every activation or termination event the link type has to be considered. Then the likelihood component of the adaptive network process is
\begin{equation*}
    \alpha_{SS}^{C_{HH}}\alpha_{SI}^{C_{HI}}\alpha_{II}^{C_{II}}\omega_{SS}^{D_{HH}}\omega_{SI}^{D_{HI}}\omega_{II}^{D_{II}}\prod_{j=2}^n \tilde{M}(t_j) \exp\left(-\int_0^{T_{\max}}\left[\tilde{\alpha}^T\mathbf{M}_{\max}(t) + (\tilde{\omega}-\tilde{\alpha})^T\mathbf{M}(t)\right] dt \right),
\end{equation*}
where
\begin{align}
    \tilde{M}(t_j) =& \left[(\alpha_{SS}M_{HH}^d(t_j))^{C_j}(\omega_{SS}M_{HH}(t_j)^{D_j}\right]^{(1-g(p_{j1},t_j))(1-g(p_{j2},t_j))} \nonumber  \\
    & \times \left[(\alpha_{SI}M_{HI}^d(t_j))^{C_j}(\omega_{SI}M_{HI}(t_j))^{D_j}\right]^{\lvert g(p_{j1},t_j)-g(p_{j2},t_j)\rvert}  \\ 
    & \times \left[(\alpha_{II}M_{II}^d(t_j))^{C_j}(\omega_{II}M_{II}(t_j))^{D_j}\right]^{g(p_{j1},t_j)g(p_{j2},t_j)} \nonumber \\
    \tilde{\alpha} =& (\alpha_{SS}, \alpha_{SI}, \alpha_{II})^T,
    \label{eq:alpha-vector}\\
    \tilde{\omega} =& (\omega_{SS}, \omega_{SI}, \omega_{II})^T,
    \label{eq:omega-vector}\\
    \mathbf{M}_{\max}(t) =& \left(\frac{H(t)(H(t)-1)}{2}, H(t)I(t), \frac{I(t)(I(t)-1)}{2}\right)^T,\label{eq:Mmax-vector}\\
    \mathbf{M}(t) =& (M_{HH}(t_j), M_{HI}(t),
    M_{II}(t))^T. \label{eq:M-vector}
\end{align}
Therefore, given the initial network structure $\mathcal{G}_0$ and one infectious case at time $0$, the complete data likelihood of the coupled process can be expressed as \footnote{\added{The likelihood derived here can be slightly modified to describe an SIS-type epidemic instead; see \cref{app:sis-lik}.}}
\begin{align} 
\label{eq:SIR-comp-lik}
    &\mathcal{L}(\beta, \gamma, \tilde\alpha, \tilde\omega|\mathcal{G}_0) = p(\text{epidemic events},\text{network events}|\beta, \gamma, \tilde\alpha, \tilde\omega,\mathcal{G}_0) \nonumber \\
    =& \gamma^{n_R} \beta^{n_E - 1} \alpha_{SS}^{C_{HH}}\alpha_{SI}^{C_{HI}}\alpha_{II}^{C_{II}}\omega_{SS}^{D_{HH}}\omega_{SI}^{D_{HI}}\omega_{II}^{D_{II}} \prod_{j=2}^n \left[ \tilde{M}(t_j) \left(I_{p_{j1}}(t_j)\right)^{F_j} \right] \nonumber \\
    & \times \exp\left(-\int_0^{T_{\max}}\left[\beta S\kern-0.14em I(t) + \gamma I(t) + \tilde{\alpha}^T\mathbf{M}_{\max}(t) + (\tilde{\omega}-\tilde{\alpha})^T\mathbf{M}(t) \right]dt\right).
\end{align}

\paragraph{Inference Given Complete Event Data}
The likelihood function (\ref{eq:SIR-comp-lik}) will be used toward inference under missing data, but immediately suggests straightforward procedures when the process is fully observed. Given the complete event data $\{e_j\}_{j=1}^n$ and the initial conditions of the process $\mathcal{G}_0$ and $I(0)$, the only unknown quantities in (\ref{eq:SIR-comp-lik}) are the model parameters $\Theta = \{\beta, \gamma, \alpha_{SS},\alpha_{SI},\alpha_{II}, \omega_{SS},\omega_{SI},\omega_{II}\}$. The following Theorems state results on maximum likelihood estimation as well as Bayesian estimation.

\begin{thm}[Maximum likelihood estimation]
\label{thm:MLEs}
Following the likelihood function in (\ref{eq:SIR-comp-lik}), given $\mathcal{G}_0$ and complete event data $\{e_j\}$, the MLEs of the model parameters are given as follows:
\begin{align}
    %\label{eq:MLEs-st}
    \hat{\beta} = \frac{n_E - 1}{\sum_{j=1}^{n}SI(t_j)(t_j - t_{j-1})},\quad & 
    \hat{\gamma} = \frac{n_R}{\sum_{j=1}^{n}I(t_j)(t_j - t_{j-1})}, \nonumber \\
    \hat{\alpha}_{SS} = \frac{C_{HH}}{\sum_{j=1}^{n}\left[\frac{H(t_j)(H(t_j)-1)}{2} - M_{HH}(t_j)\right](t_j - t_{j-1})},\quad &\hat{\omega}_{SS} = \frac{D_{HH}}{\sum_{j=1}^{n}M_{HH}(t_j)(t_j - t_{j-1})},\nonumber \\
    %\label{eq:MLEs}
    \hat{\alpha}_{SI} = \frac{C_{HI}}{\sum_{j=1}^{n}\left[H(t_j)I(t_j) - M_{HI}(t_j)\right](t_j - t_{j-1})}, \quad &\hat{\omega}_{SI} = \frac{D_{HI}}{\sum_{j=1}^{n} M_{HI}(t_j)(t_j - t_{j-1})}, \nonumber \\
    \hat{\alpha}_{II} = \frac{C_{II}}{\sum_{j=1}^{n}\left[\frac{I(t_j)(I(t_j)-1)}{2} - M_{II}(t_j)\right](t_j - t_{j-1})},
     \quad &
    \hat{\omega}_{II} = \frac{D_{II}}{\sum_{j=1}^{n}M_{II}(t_j)(t_j - t_{j-1})}. \nonumber
    %\label{eq:MLEs-en}
\end{align}
\end{thm}

The above results can be directly obtained by setting all partial derivatives of the log-likelihood to zero. The detailed proof is provided in \cref{app:aux-proofs}.

\begin{thm}[Bayesian inference with conjugate priors] 
\label{thm:Bayesian} Under Gamma priors
%If Gamma priors are adopted for all the parameters:
\begin{equation*}
    \beta \sim Ga(a_{\beta}, b_{\beta}),\quad \gamma \sim Ga(a_{\gamma}, b_{\gamma}),\quad \alpha_{\cdot \cdot} \sim Ga(a_{\alpha}, b_{\alpha}),\quad \omega_{\cdot \cdot} \sim Ga(a_{\omega}, b_{\omega})
\end{equation*}
and given initial network $\mathcal{G}_0$ and complete data $\{e_j\}$, the posterior distributions of model parameters under likelihood (\ref{eq:SIR-comp-lik}) are given by
\begin{equation}
    \label{eq:posteriors}
    \begin{array}{ll}
    \beta | \{e_j\} \sim Ga(a_{\beta} + (n_E - 1), b_{\beta} + (n_E - 1)/\hat{\beta} ),  &   \gamma | \{e_j\} \sim Ga\left(a_{\gamma} + n_R, b_{\gamma}+ n_R/\hat{\gamma}\right),\\
    \alpha_{SS}| \{e_j\} \sim Ga \left( a_{\alpha}+C_{HH}, b_{\alpha}+C_{HH}/\hat{\alpha}_{SS} \right), & \omega_{SS} | \{e_j\} \sim Ga\left(a_{\omega}+D_{HH}, b_{\omega}+D_{HH}/\hat\omega_{SS}\right), \\
    \alpha_{SI}| \{e_j\} \sim Ga\left(a_{\alpha}+C_{HI}, b_{\alpha}+C_{HI}/\hat\alpha_{SI}\right), &
    \omega_{SI} | \{e_j\} \sim Ga\left(a_{\omega}+D_{HI}, b_{\omega}+D_{HI}/\hat\omega_{SI}\right), \\
    \alpha_{II}| \{e_j\} \sim Ga\left(a_{\alpha}+C_{II}, b_{\alpha}+C_{II}/\hat\alpha_{II}\right),
     & 
    \omega_{II} | \{e_j\} \sim Ga\left(a_{\omega}+D_{II}, b_{\omega}+D_{II}/\hat\omega_{II}\right),
    \end{array}
\end{equation}
where $\hat{\beta}, \hat{\gamma}, \hat{\alpha}_{SS}, \hat\alpha_{SI}, \hat\alpha_{II}, \hat\omega_{SS}, \hat\omega_{SI}, \hat\omega_{II}$ are the MLEs defined in Theorem~\ref{thm:MLEs}.
\end{thm}
\added{Equation~(\ref{eq:SIR-comp-lik}) is consistent with the general form of likelihood of a continuous-time Markov chain with Exponentially distributed dwell times.}
%Note that (\ref{eq:SIR-comp-lik}) implies that each parameter is essentially the rate of Exponential wait times between consecutive events of a continuous-time Markov process 
% \alex{looks resolved :)}
% \fbcom{(\ref{eq:SIR-comp-lik}) is the likelihood; I've rewritten this sentence though.}
% \jxadd{ \bf I was slightly confused by this sentence-- I might suggest something like a better-written version of ``(\ref{eq:SIR-comp-lik}) is consistent with the general form for MLEs of a CTMC of count sufficient statistics normalized by dwell times".} 
Applying the Gamma-Exponential conjugacy leads to the posterior distributions in (\ref{eq:posteriors}).

\paragraph{Relaxing the closed population assumption} 
% Throughout this section, 
% added by Alex
\added{Above, }the host population \added{was} %is 
implicitly assumed to be closed by fixing $N$. If in reality the observed population of size $N$ is a subset of a larger unobserved population, then it is possible for an individual to get infected by an external source. % that is not under study. 
\added{This is not reflected in the likelihood above as there is no corresponding $S$-$I$ link  within the observed population}, %link infection rate $\beta$, 
so we introduce an ``external infection'' rate $\xi$ describing the rate for each susceptible individual to contract the disease from an external source. In other words, $\xi$ can be thought of as the constant rate
for any $S$ individual to enter status $I$ independently of interaction with infectious members in the observed population. In this scenario, the complete data likelihood becomes
%but a slight model modification can resolve this issue. Let $\xi$ be the ``external infection'' rate, the rate for any susceptible individual to be infected by any external infectious source, then the complete data likelihood is
\begin{align} 
\label{eq:SIR-lik-external}
    & \mathcal{L}(\beta, \xi, \gamma, \tilde\alpha, \tilde\omega|\mathcal{G}_0) %= p(\text{epidemic events},\text{network events}|\beta, \xi, \gamma, \alpha, \omega,\mathcal{G}_0) \nonumber \\
    = \gamma^{n_R}
    \alpha_{SS}^{C_{HH}}\alpha_{SI}^{C_{HI}}\alpha_{II}^{C_{II}}\omega_{SS}^{D_{HH}}\omega_{SI}^{D_{HI}}\omega_{II}^{D_{II}} \prod_{j=2}^n \left[ \tilde{M}(t_j) \left(\beta I_{p_{j1}}(t_j) + \xi\right)^{F_j} \right] \nonumber  \\
    & \times \exp\left(-\int_0^{T_{\max}}\left[\beta S\kern-0.14em I(t) + \xi S(t) + \gamma I(t) + \tilde{\alpha}^T\mathbf{M}_{\max}(t) + (\tilde{\omega}-\tilde{\alpha})^T\mathbf{M}(t) \right]dt\right). \nonumber
\end{align}
The MLEs for $\{\gamma, \tilde{\alpha}, \tilde{\omega}\}$ remain unchanged, and though there is no longer a closed-form solution to the MLEs for $\beta$ and $\xi$, numerical solutions can be easily obtained \added{as detailed in  \cref{app:subpopulation-inference}}. % by any commonly used nonlinear optimization algorithm
In particular, if we have information on which infection cases are caused by internal sources (described by $\beta$) and which are caused by external sources (described by $\xi$), we can directly obtain the MLEs (and posterior distributions) for all the parameters. In this case, estimation for all parameters except $\beta$ and $\xi$ remains unchanged. When there is missingness in recovery times, the Bayesian inference procedure proposed in the next section can still be carried out with only minor adaptations. 
%More details about inference with observations on an open population are included in \cref{app:subpopulation-inference}.

\section{Inference with Partial Epidemic Observations}
\label{sec:missrecov}
Though likelihood-based inference is straightforward when all events are observed, complete event data are rarely collected in real-world epidemiological studies. %Thanks to recent developments in mobile device technologies, individual social contacts can be accurately and dynamically tracked, thus making any structural change in an evolving social network readily detectable. It is the epidemic histories, more precisely, the infection and recovery times, that are likely to be unavailable. 
Even in epidemiological studies with very comprehensive observations %(e.g. the study by 
\citep{aiello2016design}, % in which individual social contacts were dynamically tracked using mobile devices), 
there still exists some degree of missingness in the exact individual recovery times. 
\added{In these data, infection times are recorded when a study subject reports symptoms, but recoveries are aggregated at a coarse time scale rather than % almost never 
immediately recorded when a subject becomes disease-free}. 
% \footnote{\fbcom{Here, for simplicity, we assume that an infected individual always exhibits symptoms, and that once an individual recovers from the symptoms he is disease-free and thus no longer infectious.}}

Incomplete observations on epidemic paths have long been a major challenge for inference, even assuming a randomly mixing population or a simple, fixed network structure. With the exact times of infections and/or recoveries unknown, it essentially requires integrating over all possible individual disease episodes to obtain the \textbf{marginal} likelihood % posterior distribution 
of the parameters. In most cases, this is intractable. Instead, our strategy is to bypass the direct marginalization through data augmentation. This entails treating the unknown quantities in data as latent variables, iteratively imputing their values, and then estimating parameters given the current setting of latent variables. By \textbf{augmenting} the data via latent variables, the parameter estimation step makes use of the computationally tractable complete-data likelihood. This class of methods \added{has proven successful} for related problems based on individual-level data \citep{auranen2000transmission, hohle2002estimating, cauchemez2006s, hoti2009outbreaks,Tsang19vaccine}, % \fbcom{(added a reference per Allison's recommendation)},
population-wide prevalence counts \citep{fintzi2017efficient}, and observations on a structured but static population \citep{neal2004statistical,o2009bayesian,Tsang19vaccine}, but has not yet been designed for an epidemic process coupled with a \textbf{dynamic} network. The time-varying nature of social interactions imposes complex constraints on the data augmentation. % which already involves difficult conditioned sampling, so previous techniques can be quite ineffective or inefficient. 
\added{Though network dynamics complicate the design of data-augmented samplers, the information they provide on possible infection sources and transmission routes allow us to exploit additional structure, effectiveley reducing the size of the latent space.} %that should be exploited.
%can lead to complex constraints for proposing feasible disease episodes but at the same time provides information on possible infection sources and transmission routes, which renders previous data augmentation techniques ineffective or inefficient. 

%In this section, 
We derive a data augmentation method specifically designed to enable inference under  missing recovery times. The algorithm utilizes the information presented by the dynamic contact structure. In contrast to existing methods (such as \cite{fintzi2017efficient} and \cite{hoti2009outbreaks}), it is able to efficiently impute unobserved event times in parallel instead of updating individual trajectories one by one. \added{We focus on the case when recovery times are missing, directly motivated by the case study data in Section~\ref{sec:real-experiments}, but note that the proposed framework applies to other sources of missing data; see Section~\ref{sec:discussion} for discussion.}

\subsection{Method Overview}
\label{sec:missing-overview}
\paragraph{Problem setting}
% resolved %
%\jxadd{I would again recap what the data setting is for the dataset at hand, before contrasting the complete data and the given data in detailed notation}/

%We would like to resolve the issue of missing recovery times in the observed data. That is, the fact that someone makes recovery is known but the exact time of such an event is unavailable.
\ifthenelse{\boolean{tocut}}{Throughout the observation period $(0, T_{\max}]$, suppose $\{(u_{\ell},v_{\ell}]\}_{\ell=1}^L$ ($u_{\ell} < v_{\ell}$ and $v_{\ell} \leq u_{\ell+1}$) is the collection of \textbf{disjoint} time intervals in which a certain number of recoveries occur, but the exact times of those recoveries are unknown.}{
As introduced in Section~\ref{sec:model}, the complete event data $\{e_j = (t_j, p_{j1},p_{j2}, F_j, C_j, D_j)\}$ consist of the event times ($t_j$'s), identities of the individuals involved ($p_{j1},p_{j2}$'s), and the types of each event (labeled by indicators $F_j, C_j, D_j$'s). Now, suppose that the event times of recoveries are partially (or fully) missing, that is,
\begin{equation*}
    \text{for some $j$ with } F_j = C_j = D_j = 0,\, t_j \text{ is unknown. } 
\end{equation*}

Further assume that we have access to a series of disease status reports at certain time points during the observation period $(0, T_{\max}]$. At time $t$, for example, one report contains the disease status (i.e. indicator of ill or healthy) for every member of the population. Access to such information is usually granted in epidemiological studies where the health status of every study subject is periodically reported through regular surveys (for example, weekly surveys).

Let $\{(u_{\ell},v_{\ell}]\}_{\ell=1}^L$ ($u_{\ell} < v_{\ell}$ and $v_{\ell} \leq u_{\ell+1}$) be the collection of \textbf{disjoint} time intervals in which a certain number of recoveries occur, but the exact times of those recoveries are unknown.} 
That is, for each $\ell = 1,2,\ldots,L$, some individuals are reported as infectious up to time $u_{\ell}$, and they are reported as healthy again starting from time $v_{\ell}$. 
Within one particular interval $(u_{\ell},v_{\ell}]$, let $n_E^{(\ell)}$ be the number of infections, and $n_R^{(\ell)}$ the number of recoveries for which the \textbf{exact times are known}, so the number of \textbf{unknown} recovery times for this interval is $R_{\ell} = I(u_{\ell})-I(v_{\ell})+n_E^{(\ell)}-n_R^{(\ell)}$.
%\footnote{Assume that the interval $(u_{\ell},v_{\ell}]$ is reasonably short such that no one can get infected after time $u_{\ell}$ \textbf{and} recover before time $v_{\ell}$.} 
%Label the $R_{\ell}$ individuals with unknown recovery times \added{(inside the $\ell$th interval)} by $ \{k_{\ell,1},\ldots,k_{\ell,R_{\ell}}\}$, and d
\added{Denote these recovery times by latent variables} \added{$\mathbf{r}_{\ell} =\{r_{\ell,1},\ldots, r_{\ell,R_{\ell}}\}$. }
\added{Our goal is to conduct inference despite the absence of all the exact recovery times $\mathbf{r} = \{\mathbf{r}_{\ell}\}_{\ell=1:L}$ in the observed data. }

Further assume that we have a health status report (indicating ill or healthy) of each individual periodically during $(0, T_{\max}]$. Access to such information is usually granted in epidemiological studies where every study subject gives updates on health statuses through regular surveys (for example, weekly surveys).

\paragraph{Inference scheme}
We propose to address the problem of missing recovery times through data-augmented Markov chain Monte Carlo. Given an initial guess of parameter values $\Theta^{(0)}$ and the observed data $\mathbf{x}=\{e_j\} \cup \{\text{health status reports}\} \cup \mathcal{Z}_0$, \added{at each iteration $s$} the algorithm samples a set of values for the missing recovery times \added{$\mathbf{r}^{(s)} = \{\mathbf{r}^{(s)}_{\ell}\}_{\ell=1:L}$} %$\mathbf{r}^{(s)} = \{r_{\ell,1},\ldots, r_{\ell,R_{\ell}}\}_{\ell=1:L}^{(s)}$ 
from their probability distribution conditioned on $\mathbf{x}$ and the current draw of parameter values. It then samples a new set of parameter values $\Theta^{(s)}$ from their posterior distributions conditioned on the augmented data. In summary, for $s=1:S$ where $S$ is the maximum iteration count:
\vspace*{-0.1in}
\begin{enumerate}
    \item \textbf{Data augmentation}. Draw \added{$\mathbf{r}^{(s)} = \{\mathbf{r}^{(s)}_{\ell}\}_{\ell=1:L}$} from the joint conditional distribution
    % \begin{equation}
    % \label{eq:cond-recov-orig}
    %     p(\added{\mathbf{r}_{\ell}}|\Theta^{(s-1)},\mathbf{x}, \{r_{\ell',i}\}_{i=1:R_{\ell'},\ell'\neq \ell}) \qquad \text{for } \ell = 1:L.
    % \end{equation}
    \begin{equation}
    \label{eq:cond-recov-orig}
        p(\added{\mathbf{r}}|\Theta^{(s-1)},\mathbf{x}, \mathbf{r}^{(s-1)}).
    \end{equation}
    % \alex{hm, so the above is not actually exactly right---there should be a super script s-1 on some of the rl and an s on some of them... that's a bit of a pain to describe though, my attempt is below which is terrible...}
    % \begin{equation}
    % \label{eq:cond-recov-orig}
    %     p(\added{\mathbf{r}_{\ell}}|\Theta^{(s-1)},\mathbf{x}, \mathbf{r}_{-\ell}^{(s-1)}) \qquad \text{for } \ell = 1:L
    % \end{equation}
    % where $\mathbf{r}_{-\ell}^{(s-1)} = \{r^{(s)}_{\ell',i}\}_{i=1:R_{\ell'},\ell'< \ell}\cup\{r^{(s-1)}_{\ell',i}\}_{i=1:R_{\ell'},\ell'> \ell}$. \fbcom{The update doesn't have to happen in such a sequential way, though. Is it okay to write, say, $\mathbf{r}_{-\ell}^{(s-1)} = \{\mathbf{r}_{\ell'}^{(s-1)}\}_{\ell' \neq \ell}$ instead? i.e., Is it okay to update the $\mathbf{r}_{\ell}$'s in iteration $s$ only conditioned the samples acquired in iteration $s-1$?} 
    \item \textbf{Parameter update}. Combine $\mathbf{x}$ and \added{$\mathbf{r}^{(s)}$} %$\{r_{\ell,i}\}_{i=1:R_{\ell},\ell=1:L}^{(s)}$ 
    to form the augmented, complete data. Sample parameters $\Theta^{(s)} \mid \mathbf{x}, \mathbf{r}^{(s)}$ according to %Eq.(\ref{eq:posterior-st})-(\ref{eq:posterior-en}).
    (\ref{eq:posteriors}).
\end{enumerate}

\subsection{Data Augmentation via Endpoint-conditioned Sampling}
In the inference scheme stated above, the data augmentation step (step 1) is challenging because (\ref{eq:cond-recov-orig}) describes the distribution of missing recovery times conditioned on both historical events \textbf{and} future events. Thus drawing from (\ref{eq:cond-recov-orig}) amounts to sampling unobserved event times from a continuous-time Markov process with a series of fixed endpoints \citep{hobolth2009simulation}, a challenging task.  Even though (\ref{eq:recovery-background}) suggests that, in \textbf{forward simulations}, the time it takes for an infectious person to recover only depends on the recovery rate $\gamma$,% and is completely independent of the individual's social links or the epidemic history of any other individual, 
when recovery times need to be inferred \textbf{retrospectively}, there are additional constraints imposed by the observed data. % related to social links and the epidemic history of other individuals. 
First, an individual $q$ cannot recover before a certain time point $t$ if it is observed that at time $t$ the person is still ill. More subtly, if another individual $p$ gets infected during his contact with $q$, then the associated recovery time for $q$ cannot leave $p$ without a possible infection source. 
The first condition is easy to satisfy. The second constraint is much more complicated due to the network dynamics, which a simple forward simulation approach would fail to effectively accommodate.

We tackle the challenge in data augmentation by first simplifying the expression of  (\ref{eq:cond-recov-orig}) and then stating an efficient sampling algorithm.

\begin{lem}
\label{lem:cond-recov}
(\ref{eq:cond-recov-orig}) can be simplified into the following expression:
\begin{equation}
    \label{eq:cond-recov-simp}
       \added{\prod_{\ell = 1:L} p\left(\mathbf{r}_{\ell}|\gamma^{(s-1)},\{e_j\}_{t_j\in(u_{\ell},v_{\ell}]}, \mathcal{Z}_{u_{\ell}}\right),} %\{\text{epidemic status at time } v_{\ell}\}.
\end{equation}
where $\mathcal{Z}_{t}$ is the state of the process at time $t$, including the epidemic status of each individual and the social network structure.
\end{lem}
\begin{proof}
%\jxadd{\bf Can we use the vector $\mathbf{r}$ to shorten/lighten notation? It is factored over $l$ so maybe not fruitful..} \fbcom{Great idea! I've made the changes. } 
Consider the joint density of the complete data given parameter values $\Theta^{(s-1)}$.
\begin{align*}
    &p\left(\mathbf{x}, \added{\{\mathbf{r}_{\ell}\}_{\ell=1:L}}|\Theta^{(s-1)}\right)  \\
    =& \prod_{\ell=1:L} \left[p\left(\{e_j\}_{t_j\in(u_{\ell},u_{\ell+1}]},\added{\mathbf{r}_{\ell}}|\mathcal{Z}_{u_{\ell}} ,\Theta^{(s-1)}\right)\right] \times p\left(\{e_j\}_{t_j \leq u_{1} \text{ or } t_j > v_{L}}|\mathcal{Z}_0,\mathcal{Z}_{v_L},\Theta^{(s-1)}  \right)  \\
    =& \prod_{\ell=1:L} \left[p\left(\{e_j\}_{t_j\in(u_{\ell},v_{\ell}]}|\added{\mathbf{r}_{\ell}},\mathcal{Z}_{u_{\ell}} ,\Theta^{(s-1)}\right)p\left(\added{\mathbf{r}_{\ell}}|\mathcal{Z}_{u_{\ell}}, \gamma^{(s-1)}\right)\right]  \\ 
    & \times \left[\prod_{\ell=1:L} p\left(\{e_j\}_{t_j\in(v_{\ell},u_{\ell+1}]}|\mathcal{Z}_{v_{\ell}},\Theta^{(s-1)}\right)\right] p\left(\{e_j\}_{t_j \leq u_{1} \text{ or } t_j > v_{L}}|\mathcal{Z}_0,\mathcal{Z}_{v_L},\Theta^{(s-1)}  \right).
\end{align*}
Examining all terms concerning $\added{\mathbf{r}_{\ell}}$ for each $\ell$ 
\added{indicates that, conditioned on $\gamma^{(s-1)}$, $\{e_j\}_{t_j\in(u_{\ell},v_{\ell}]}$, and $ \mathcal{Z}_{u_{\ell}}$, the distribution of $\mathbf{r}_{\ell}$ does \textbf{not} depend on $\{\mathbf{r}_{\ell'}\}_{\ell' \neq \ell}$. Thus,}
%leads to \alex{below has the same issue as (17) though I don't know that what I did there is the correct solution}
%For every $\ell = 1,2,\ldots, L$, examine the terms , we can obtain that 
\begin{equation*}
    %p\left(\added{\mathbf{r}_{\ell}}|\Theta^{(s-1)},\mathbf{x}, \added{\{\mathbf{r}_{\ell'}\}_{\ell' \neq \ell}}\right) = p\left(\added{\mathbf{r}_{\ell}}|\gamma^{(s-1)},\{e_j\}_{t_j\in(u_{\ell},v_{\ell}]}, \mathcal{Z}_{u_{\ell}}\right).
    \added{p(\added{\mathbf{r}}|\Theta^{(s-1)},\mathbf{x}, \mathbf{r}^{(s-1)}) = \prod_{\ell = 1:L} p\left(\added{\mathbf{r}_{\ell}}|\gamma^{(s-1)},\{e_j\}_{t_j\in(u_{\ell},v_{\ell}]}, \mathcal{Z}_{u_{\ell}}\right).}
\end{equation*}
\end{proof}

The lemma above suggests that imputation of missing recovery times inside an interval $(u,v]$ only depends on the events that occur in $(u,v]$, the state of the process at the start of the interval, $\mathcal{Z}_u$, and the value of recovery rate $\gamma$. \added{Further, imputation on disjoint intervals can be conducted separately and in parallel.}

Now consider sampling recovery times within any interval $(u,v]$. Let $\mathcal{Q}$ denote the group of individuals who recover at unknown times during $(u,v]$, and for each $q \in \mathcal{Q}$, let $q$'s exact recovery time be $r_q$ ($\in (u,v]$). Similarly, let $\mathcal{P}$ denote the group of individuals who get infected during $(u,v]$; for $p \in \mathcal{P}$, let $p$'s infection time be $i_p$, $\mathcal{N}_p(i_p)$ be the set of $p$'s contacts at time $i_p$, and $\mathcal{I}(i_p)$ be the set of \textbf{known} infectious individuals at time $i_p$ (that is, $\mathcal{I}(i_p)$ excludes any individual who may have recovered before $i_p$). 

\begin{figure}[H]
    \centering
    %\vspace*{-0.1in}
    %\includegraphics[width=\textwidth,trim={0 0 0 0.4in},clip]{figures/Chewbacca_portrait.pdf}
    \includegraphics[width=\textwidth,trim={0 0 0 0.4in},clip]{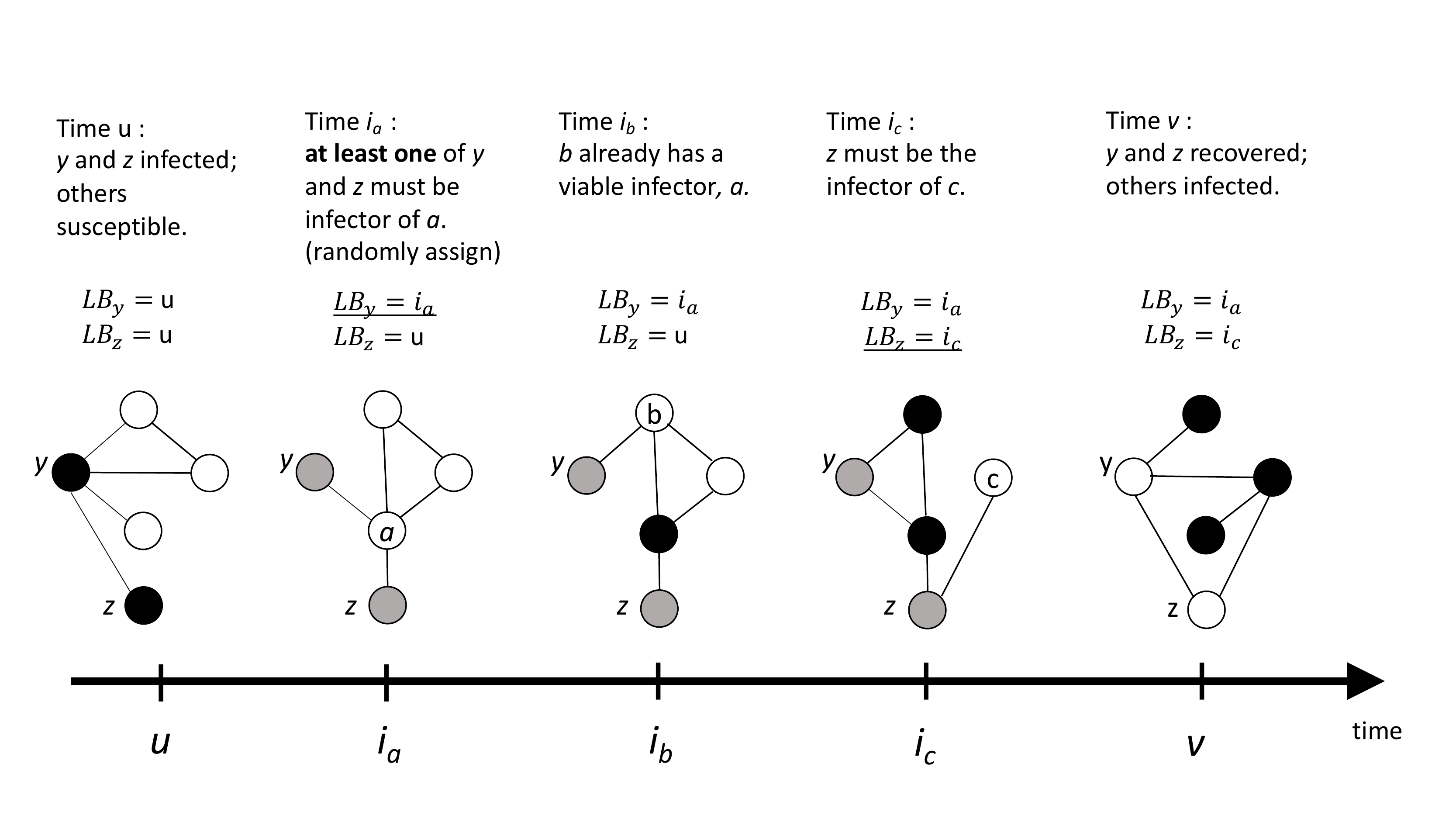}
    \caption{\added{Illustration of DARCI on a $N=5$ population. Each circle represents an individual and each solid line represents a link. Disease status is color-coded: dark $=$ infectious, grey $=$ unknown (possibly infectious or recovered), and white $=$ healthy (susceptible or recovered). Individuals $y$ and $z$ are known to be infectious at time $u$ but are recovered by time $v$, and individuals $a$, $b$ and $c$ are known to get infected at time points $i_a$, $i_b$ and $i_c$, respectively. For each person $p \in \{a,b,c\}$, the DARCI algorithm inspects $p$'s contacts at infection time $i_p$, and updates ``lower bounds'' ($LB$) of $y$ and $z$'s recovery times to ensure that $p$ has an infector. For example, at time $i_a$, one of $y$ and $z$ has to be $a$'s infector, so DARCI randomly selects \textbf{one} of $y$ and $z$ (in this example it's $y$) and postpones his recovery time until after $i_a$.}}
    \label{fig:DARCI-portrait}
\end{figure}

%The following proposition states an algorithm for sampling $\{r_q\}$'s.
\begin{prop}[Data augmentation regulated by contact information (DARCI)] 
\label{prop:DARCI-alg}
%Let $(u,v]$ be a time interval during which a group of individuals recover from the infection but the exact recovery times are unknown. Denote this group by $\mathcal{Q}$, and for each $q \in \mathcal{Q}$, let $q$'s recovery time be $r_q$ ($\in (u,v]$). Furthermore, let $\mathcal{P}$ denote the group of individuals who get infected during $(u,v]$, and for $p \in \mathcal{P}$, let $p$'s infection time be $i_p$. Then, 
Following the notation stated above, given a recovery rate $\gamma$, the state of the process at time $u$, $\mathcal{Z}_u$, and all the observed events in the interval $(u,v]$, $\{e_j\}_{t_j \in (u,v]}$, one can sample $\{r_q\}$ from the conditional distribution $p\left(\{r_q\}|\gamma, \{e_j\}_{t_j \in (u,v]},\mathcal{Z}_u\right)$ in the following steps:
\begin{enumerate}
        \item Initialize a vector $\text{LB}$ of length $|\mathcal{Q}|$ with $\text{LB}_q = u$ for every $q \in \mathcal{Q}$; \added{then for any $p \in \mathcal{P}$ such that $p \in \mathcal{Q}$, further set $\text{LB}_p = u$;}
        \item Arrange the set $\mathcal{P}$ in the order of $\{p_1, p_2, \ldots, p_{|\mathcal{P}|}\}$ such that $i_{p_1} < i_{p_2} < \ldots < i_{p_{|\mathcal{P}|}}$, and for each $p \in \mathcal{P}$ (chosen in the arranged order), examine the \added{``potential infectious neighborhood''}
        \begin{equation*}
            \added{\mathcal{I}_p = \mathcal{N}_p(i_p) \cap \left(\added{\mathcal{I}(i_p)} \cup \mathcal{Q}\right).}
            %\mathcal{I}_p = (\mathcal{N}_p(i_p) \cup \{p\}) \cap \left(\added{\mathcal{I}(i_p)} \cup \mathcal{Q}\right).
        \end{equation*}
        %Here $\mathcal{N}_p(t)$ is the set of $p$'s neighbors at time $t$, and $\text{Inf}(t)$ is the set of \textbf{known} infected individuals at time $t$ (that is, if person $k$ who is infected before $t$ \textbf{may or may not} recover before $t$, then $k$ is not included).\\ 
        If $\mathcal{I}_p \subset \mathcal{Q}$ \added{(i.e., potential infection sources are all members of $\mathcal{Q}$)}, then randomly and uniformly select one $q \in \mathcal{I}_p$, and set $\text{LB}_q = i_p$.
        \item Draw recovery times $r_q  \stackrel{ind}{\sim} \text{TEXP}(\gamma, \text{LB}_q , v)$, where $\text{TEXP}(\gamma, s, t)$ is a truncated Exponential distribution with rate $\gamma$ and truncated on the interval $(s,t)$.
    \end{enumerate}
\end{prop}
Intuitively, this procedure enables a draw of recovery times that are ``consistent with'' the observed data. To achieve this goal, an imputed recovery cannot occur in a way that leaves a to-be-infected individual without any infectious neighbor at the time of infection, nor take place before the corresponding individual gets infected. Effectively there is a ``lower bound'' for each missing recovery time conditioned on the observed data, particularly the dynamic contact structure. An illustration of the DARCI algorithm is provided in Figure~\ref{fig:DARCI-portrait}.

Combining Lemma~\ref{lem:cond-recov} and Proposition~\ref{prop:DARCI-alg} enables exact sampling from the conditional distribution (\ref{eq:cond-recov-simp}) in the data augmentation step: for each $\ell=1,2,\ldots,L$, applying the DARCI algorithm to the interval $(u_{\ell},v_{\ell}]$ gives an updated set of missing recovery times, \added{$\mathbf{r}_{\ell}=\{r_{\ell,i}\}_{i=1:R_{\ell}}$}. \added{This allows us to carry out MCMC sampling using a simple Gibbs sampler. }

\section{Simulation Experiments}
\label{sec:sim-experiments}

%\jxadd{Experimental results look great and we can later decide how many plots to include and what points to emphasize. One thing we should try is to compare fitting the regular well-mixed SIR model to the data generated from the dynamic network processes, and show how it breaks down as the connectivity increases while our method rmeains reliable. Additionally, we may see if we can compare fitting our own model in the partially observed case using generic ``plug-and-play" simulation approaches. These are basically particle filtering (for instance R package POMP), but I imagine they will deteriorate and also be slower.}

%\fbcom{Not sure if we still want to try out simulation-based approaches like particle filtering. To be addressed. -Fan}

%\fbcom{Some big-picture-style intro for the simulation experiments, summarizing the main points. To be added.}

In this section we present results of a series of experiments with simulated datasets. \ifthenelse{\boolean{tocut}}{
In all experiments, 
}{
In Section~\ref{sec: simulation-procedure}, we describe the simulation procedure we employ for generating events of adaptive network epidemic processes. In Section~\ref{sec:infer-complete}, we evaluate and validate the model and the likelihood via experiments on complete event data.
%demonstrate the necessity of considering network structures as well as network dynamics when estimating epidemic-related parameters through a small-scale example, then assess the validity of our claims on likelihood-based estimation, and show that the model is flexible enough to handle large-scale populations and arbitrary network structures. 
Finally in Section~\ref{sec:infer-incomplete}, we verify the efficacy and efficiency of the data-augmentation-based inference scheme with partial epidemic observations.

\subsection{Simulation procedure}
\label{sec: simulation-procedure}
In all the simulation experiments, }we employ a forward simulation procedure that can be seen as a variation of Gillespie's algorithm \citep{gillespie1976general} to sample realizations \added{of the network epidemic from our generative model}. %of the dynamic network epidemic process. 
The input consists of the parameter values $\Theta = \{\beta, \gamma, \tilde{\alpha}, \tilde{\omega}\}$ \footnote{Here $\tilde{\alpha} = (\alpha_{SS},\alpha_{SI},\alpha_{II})^T$ and $\tilde{\omega} = (\omega_{SS},\omega_{SI},\omega_{II})^T$, as defined in ~(\ref{eq:alpha-vector}) and (\ref{eq:omega-vector}).}, an arbitrary initial network $\mathcal{G}_0$, the number of infectious cases at onset $I(0)$, and the observation time length $T_{\max}$. The output is the complete collection of all events $\{e_j = (t_j, p_{j1},p_{j2}, F_j, C_j, D_j)\}$ that occur within the time interval $(0, T_{\max}]$. Associated with each event $e_j$ is a timestamp ($t_j$), labels of the individuals involved ($p_{j1},p_{j2}$), and the event-type indicator $F_j, C_j$, or $D_j$.

The steps of the simulation procedure are detailed as follows:
\vspace*{-0.08in}
\begin{enumerate}
    \item \textbf{Initialization.} Randomly select $I(0)$ individuals to be the infectious (then the rest of the population are all susceptible). Set $t_\text{cur}=0$.
    \item \textbf{Iterative update.} While $t_\text{cur} < T_\text{max}$, do:
    \begin{enumerate}
        \item \textbf{Bookkeeping.}
        Summarize the following statistics at $t_\text{cur}$:\ifthenelse{\boolean{tocut}}{
        1) $SI(t_\text{cur})$, the number of S-I links in the population; 2) $\mathbf{M}_{\max} (t_\text{cur})$, the possible number of links of each type defined in (\ref{eq:Mmax-vector}); 3) $\mathbf{M} (t_\text{cur})$, the number of existing links of each type defined in (\ref{eq:M-vector}).}{
        \begin{itemize}
            \item $SI(t_\text{cur})$, the number of S-I links in the population;
            \item $\mathbf{M}_{\max} (t_\text{cur})$, the possible number of links of each type defined in (\ref{eq:Mmax-vector});
            \item $\mathbf{M} (t_\text{cur})$, the number of existing links of each type defined in (\ref{eq:M-vector}).
        \end{itemize}
        }
        Then set $\mathbf{M}^{d}(t_\text{cur})=\mathbf{M}_{\max} (t_\text{cur}) - \mathbf{M} (t_\text{cur})$.
        \item \textbf{Next event time.} Compute the instantaneous rate of the occurrence of any event, $\Lambda(t_\text{cur}) = 
            \beta SI(t_\text{cur}) + \gamma I(t_\text{cur}) +
            \tilde{\alpha}^T \mathbf{M}^{d}(t_\text{cur}) +
            \tilde{\omega}^T \mathbf{M} (t_\text{cur})$,
        % \begin{equation*}
        %     \Lambda(t_\text{cur}) = 
        %     \beta SI(t_\text{cur}) + \gamma I(t_\text{cur}) +
        %     \tilde{\alpha}^T \mathbf{M}^{d}(t_\text{cur}) +
        %     \tilde{\omega}^T \mathbf{M} (t_\text{cur}),
        % \end{equation*}
        and draw $\Delta t \sim \text{Exponential}(\Lambda(t_\text{cur}))$.
        \item \textbf{Next event type.} Sample $Z \sim \text{Multinomial}(\tilde{\lambda}(t_\text{cur}))$, where
        \begin{equation*}
             \tilde{\lambda}(t_\text{cur}) = \left(\frac{\beta SI(t_\text{cur})}{\Lambda(t_\text{cur})},
            \frac{\gamma I(t_\text{cur}) }{\Lambda(t_\text{cur})},
            \frac{\tilde{\alpha}^T \mathbf{M}^{d}(t_\text{cur})}{\Lambda(t_\text{cur})},
            \frac{\tilde{\omega}^T \mathbf{M} (t_\text{cur})}{\Lambda(t_\text{cur})}\right)^T.
        \end{equation*}
        Then do one of the following based on the value of $Z$:
        
        If $Z=1$ (infection), uniformly pick one $S$-$I$ link and infect the $S$ individual in this link.
        
        If $Z=2$ (recovery), uniformly pick one $I$ individual to recover.
        
        If $Z=3$ (link activation), randomly select $Y \in \{H\text{-}H, H\text{-}I,I\text{-}I\}$ with probabilities proportional to $\tilde{\alpha} \circ \mathbf{M}^{d}(t_\text{cur})$, and
        uniformly pick one de-activated ``$Y$ link'' to activate.
        
        If $Z=4$ (link termination), randomly select $Y \in \{H\text{-}H, H\text{-}I,I\text{-}I\}$ with probabilities proportional to $\tilde{\omega} \circ \mathbf{M}(t_\text{cur})$, and
        uniformly pick one existing ``$Y$ link'' to terminate.
        \item Replace $t_\text{cur}$ by  $t_\text{cur} + \Delta t$, record relevant information about the sampled event, and repeat from (a).
    \end{enumerate}
\end{enumerate}
In Step 2 (c), ``$\circ$'' refers to the Hadamard product (entrywise product) for two vectors.

\subsection{Experiments with Complete Observations}
\label{sec:infer-complete}

In this subsection, we first demonstrate the insufficiency of analyzing disease spread without considering the network structure or its dynamics. Then we validate our claims on maximum likelihood estimation and Bayesian inference given complete event data (Theorems \ref{thm:MLEs} and \ref{thm:Bayesian}). Finally, we show that the model estimators can detect simpler models such as the decoupled process and the static network process. 
Unless otherwise stated, \added{throughout this section} we set the initial network $\mathcal{G}_0$ as a random Erdős–Rényi graph\footnote{\added{We note that the form of the initial network does not necessarily predict the behavior of the epidemic. Specifically, asymptotic qualities, such as the Poisson degree distribution of Erdős–Rényi graphs, do not hold when the network dynamically reacts to an epidemic, detailed empirically  in Section~\cref{app:network-behavior}.}} (undirected) with edge probability $p=0.1$, let $I(0)=1$ individual to get infected at onset, and choose the ground-truth parameters as
\begin{equation}
\label{eq:settings}
    \beta = 0.03, \gamma = 0.12; \tilde{\alpha}^T = %(\alpha_{SS},\alpha_{SI},\alpha_{II}) =
    (0.005, 0.001, 0.005), \tilde{\omega}^T = %(\omega_{SS},\omega_{SI},\omega_{II}) =
    (0.05, 0.1, 0.05).
\end{equation}
\added{These settings are chosen to produce simulated data sets with a population size and event counts that are comparable to our real data example.}

For Bayesian inference, we adopt the following Gamma priors for the parameters: %\footnote{
\begin{equation}
\label{eq:priors}
    \beta \sim Ga(1, 1/0.02), \gamma \sim Ga(1, 1/0.1); \alpha_{\cdot \cdot} \sim Ga(1, 1/0.004), \omega_{\cdot \cdot} \sim Ga(1, 1/0.06).
\end{equation}
We intentionally choose prior means different from the true parameter values; experiments show that inference is insensitive to prior specifications as long as a reasonable amount of data is available. %\alex{was this in here before? and did we do robustness checks for this?}
For each parameter, 1000 posterior samples are drawn after a 200-iteration burn-in period.

\paragraph{The danger of neglecting networks or network dynamics} 

Adopting the ``random mixing'' assumption about an actually networked population can lead to severe under-estimation of the infection rate.
Erroneous estimation can also happen if contacts are in fact dynamic but are mistaken as static during inference. Table~\ref{tab:toy-example-simulation} displays the MLEs of the infection rate $\beta$ (\textbf{under true value  0.05}\added{, chosen to generate non-trivial epidemics to illustrate inference}) obtained by methods under three different assumptions regarding the network structure (assuming a dynamic network, assuming a static network, and assuming random mixing without any network). The population size is $N=50$, and results are summarized over 50 different simulated datasets.

These results make clear that neglecting the effects of the contact network, even when the quantity of interest is the disease transmission rate, is dangerously misleading. Resulting estimates that are far from the truth \added{with significantly underestimated uncertainty measures}. Incorporating the initial network structure statically throughout the process helps---the 95\% confidence interval now includes the truth---but disregarding the time-evolution of the network remains a noticeable model misspecification leading to \added{biased inference}. 

\begin{table}[H]
\centering
\caption{%The danger of neglecting network structures and/or network dynamics, showcased in simulations on a $N=50$ population. The table presents 
Maximum likelihood estimates of $\beta$, the per link infection rate (real value $0.05$), using dynamic network information, the initial static network, and no network structure (random mixing), respectively. The standard deviations as well as the 2.5\% and 97.5\% quantiles of the estimates are obtained from outcomes across 50 different simulated datasets on a $N=50$ population.}
\begin{tabular}{llll}
  \toprule
Method & dynamic network & static network & no network \\ 
  \midrule
Estimate & 0.0540 & 0.0278 & 0.00219 \\ 
Standard deviation & 0.0158 & 0.0081 & 0.000821 \\ 
  2.5\% quantile & 0.0230 & 0.00825 & 0.000614 \\ 
  97.5\% quantile & 0.0817 & 0.0553 & 0.00425 \\
   \bottomrule
\end{tabular}
\label{tab:toy-example-simulation}
\end{table}

\paragraph{Validity and efficacy of parameter estimation} 

Complete event data are generated using the simulation procedure stated in above, and maximum likelihood estimates (MLEs) as well as Bayesian estimates are obtained for parameters $\Theta = \{\beta, \gamma, \tilde{\alpha}, \tilde{\omega}\}$. Here we set the population size as $N=100$ and the infection rate as $\beta = 0.03$\added{ while keeping the other parameter values the same as stated in (\ref{eq:settings})}. 

Figure~\ref{fig:MLEs-complete-CI} shows the results of maximum likelihood estimation in one simulated dataset. The MLEs for the each parameter (dark solid line) are computed using various numbers of events, and are compared with the true parameter value (red horizontal lines). The lower and upper bounds for 95\% confidence intervals are also calculated (dashed gray lines).  Only the MLEs for parameters $\beta, \gamma, \alpha_{SS}$ and $\alpha_{SI}$ are shown, but results for all parameters are included in \cref{app:sim-complete}. Estimation is relatively accurate even when observation ends earlier than the actual process (thus leaving later events unobserved). When more events are available for inference, accuracy is improved and the  uncertainty is reduced.

Figure~\ref{fig:Bayes-complete-multi} presents the posterior sample means (solid lines) and 95\% credible bands (shades) for each parameter inferred using various numbers of events, with the true parameter values marked by bold, dark horizontal lines. The results are shown for 4 different simulated datasets (each dataset represented by a distinct color) and for parameters $\beta, \gamma, \omega_{SS}$ and $\omega_{SI}$ (complete results are in \cref{app:sim-complete}). When more events are utilized in inference, the posterior means tend to be closer to the true parameter values, while the credible bands gradually narrow down.

% added by Alex
\added{It is worth noting that the proposed inferential framework is capable of handling large-scale networks as well as arbitrary network structures.}
%It is worth noting that\deleted{ parameter estimation is unaffected by either the population size $N$ or the initial network structure $\mathcal{G}_0$. In particular,} the model is capable of \replaced{conducting inference with}{handling} large-scale networks as well as arbitrary network structures. 
Additional results with larger values of $N$ and different configurations of $\mathcal{G}_0$ are provided in \cref{app:sim-complete}.

%% MLE, G0 = ER(N=100) %%
\begin{figure}[H]
    \centering
    \includegraphics[width=0.95\textwidth]{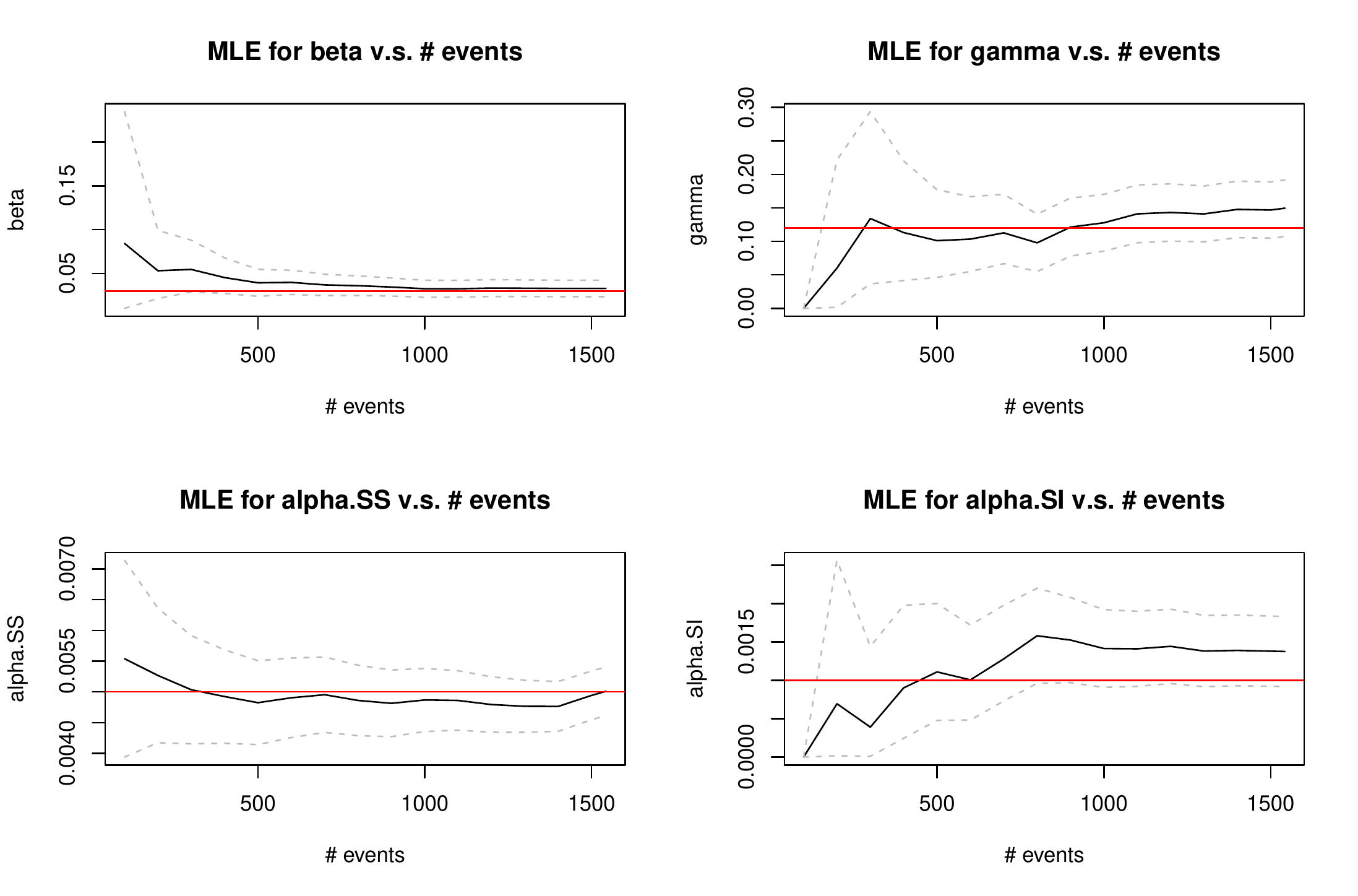}
    \caption{MLEs versus number of events used for inference. Dashed gray lines show the lower and upper bounds for 95\% frequentist confidence intervals, and red lines mark the true parameter values. Results are presented for $\beta, \gamma, \alpha_{SS}$ and $\alpha_{SI}$. \added{In this realization $n_E = n_R = 48, C_{HH}=621, C_{HI}=35, C_{II}=13, D_{HH}=573, D_{HI}=189, D_{II}=17$.}
    }
    \label{fig:MLEs-complete-CI}
\end{figure}

%% Bayesian, G0 = ER(N=100) %%
\begin{figure}[H]
    \centering
    \includegraphics[width=0.47\textwidth,page=1]{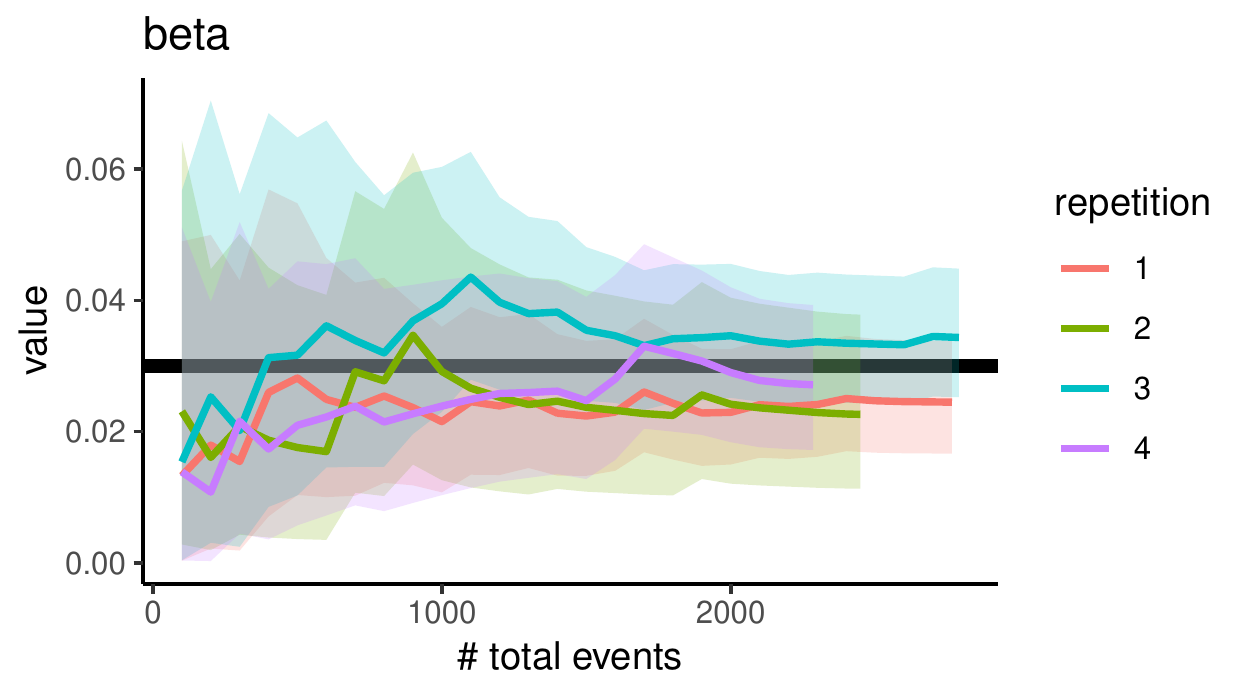}
    \includegraphics[width=0.47\textwidth,page=2]{figures/ex1_all.pdf}
    \includegraphics[width=0.47\textwidth,page=6]{figures/ex1_all.pdf}
    \includegraphics[width=0.47\textwidth,page=7]{figures/ex1_all.pdf}
    \caption{Posterior sample means versus number of total events used for inference. True parameter values are marked by \textbf{bold dark} horizontal lines, along with 95\% credible bands. Results are presented for 4 different complete datasets and for parameters $\beta, \gamma, \omega_{SS}$ and $\omega_{SI}$.}
    \label{fig:Bayes-complete-multi}
\end{figure}

\paragraph{Assessing model flexibility}
Our proposed framework is a generalization of epidemics over networks that evolve independently (the ``decoupled'' process), which in turn generalize epidemic processes over fixed networks (the ``static network'' process). \added{Thus our model class contains these simpler models:} if events are generated from the decoupled process,  we would expect all the link activation and termination rates to be estimated as the same. Likewise, if the true network process is static, then we expect all  link rates to be estimated as zero. 

To confirm this, experiments are conducted on complete event datasets generated from the two simpler models. Here we only show select results of Bayesian inference on datasets generated from static network epidemic processes (Figure~\ref{fig:Bayes-static-multi}) and relegate other results to \cref{app:sim-complete}. We see that information from a moderate number of events is sufficient to accurately estimate the epidemic parameters $\beta$ and $\gamma$ and learn the static nature of the network---note how quickly the posterior credible bands for $\alpha_{SI}$ and $\omega_{SI}$ shrink toward zero.

%% Bayesian, static network
\begin{figure}[H]
    \centering
    \includegraphics[width=0.47\textwidth,page=1]{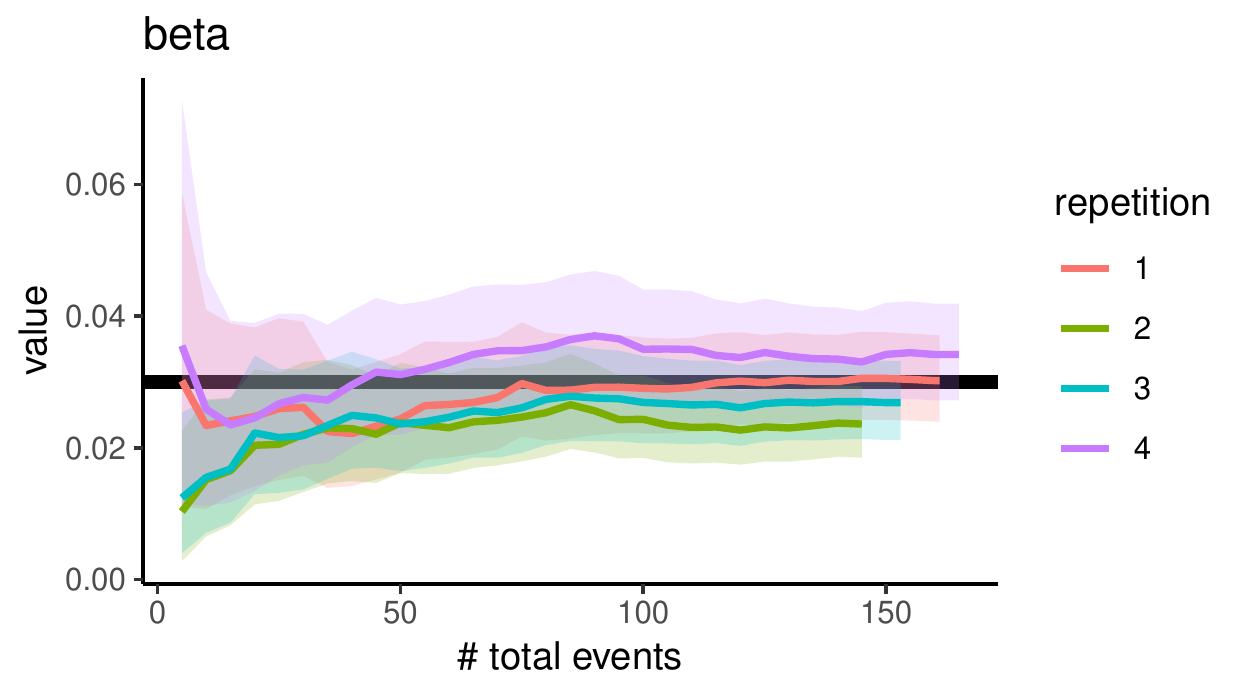}
    \includegraphics[width=0.47\textwidth,page=2]{figures/static_Bayes_1.pdf}
    \includegraphics[width=0.47\textwidth,page=4]{figures/static_Bayes_1.pdf}
    \includegraphics[width=0.47\textwidth,page=7]{figures/static_Bayes_1.pdf}
    \caption{Posterior sample means versus number of total events used for inference, with data generated from \emph{static} network epidemic processes. True parameter values are marked by \textbf{bold dark} horizontal lines, along with 95\% credible bands. Results are presented for 4 different complete datasets and for parameters $\beta, \gamma, \alpha_{SI}$ and $\omega_{SI}$.  With a moderate number of events, the epidemic-related parameters, $\beta$ and $\gamma$, are accurately estimated, and the posteriors for the edge rates quickly shrink toward zero (the truth).}
    \label{fig:Bayes-static-multi}
\end{figure}

\subsection{Experiments with Incomplete Observations}
\label{sec:infer-incomplete}
Upon validating the model and inference framework, we now  assess the performance of our proposed inference scheme in the more realistic setting where epidemic observations are incomplete.  
In this subsection, we first verify that the MCMC sampling scheme in Section~\ref{sec:missrecov} is able to retrieve the parameter values despite missing recovery times in the observed data. Then we compare our DARCI algorithm (Prop.~\ref{prop:DARCI-alg}) with two baselines and show that it produces posterior samples of higher quality and with higher efficiency.

\paragraph{Simulating partially observed data}
We first generate complete event data using the simulation procedure stated earlier in this section, and then randomly discard $\eta \times 100\%$ of the \added{\textbf{exact}} recovery times and treat them as unknown. Meanwhile, a \added{periodic} status report (as described in Section~\ref{sec:missing-overview}) is produced every $7$ time units throughout the entire process \added{so that individual disease statuses are informed at a coarse resolution}. %Note that the scale of time is insignificant here---the physical length $\tau$ of $1$ time unit can be thought of as a second, an hour, or a day, since $\tau$ only determines how we \emph{interpret} the parameters (which govern how fast the system evolves per unit time) but not the realization of the process. 
If one regards $1$ time unit as \emph{a day}, the periodical disease statuses correspond to \emph{weekly} reports.

\paragraph{Efficacy of the inference scheme} We validate the method outlined in Section~\ref{sec:missing-overview} through experiments on an example dataset, where the settings and parameters are the same as those in (\ref{eq:settings}) and the population size is fixed at $N=100$. %\footnote{These settings are chosen to generate epidemic and network events similar in scale to the real-world data to be used in the real data analysis later in the paper.}
In this particular realization, there are $26$ infection cases spanning over approximately $37$ days (less than 6 weeks), and there are $767$ and $893$ instances of social link activation and termination, respectively. \footnote{The event time scales in the example dataset are chosen to be comparable to, though not exactly the same as, those in the real-world data used in Section~\ref{sec:real-experiments}.}

First set $\eta=50$, that is, randomly select $50\%$ of \added{exact} recovery times to be taken as missing.  Figure~\ref{fig:Bayes-miss6-50} plots %\footnote{The posterior samples for the other parameters are recorded but not presented here due to limited space. Comprehensive results are included in the supplement.} 
$1000$ consecutive MCMC samples (after a $200$-iteration burn-in) \added{for each parameter  $\{\beta,\gamma,\alpha_{SS},\omega_{SS}\}$}, as well as the $2.5\%$ and $97.5\%$ quantiles of the posterior samples (gray, dashed lines) compared with the true parameter value (red horizontal line). We can see that for every parameter, the 95\% sample credible interval covers the true parameter value, suggesting that the proposed inference scheme is able to estimate parameters from incomplete data reasonably well. 

We then set $\eta = 100$, \added{discarding all exact recovery times}. Figure~\ref{fig:Bayes-miss6-100} presents outcomes of the inference algorithm in this case. Understandably, parameter estimation is affected by the total unavailability of exact recovery times, but the drop in accuracy is marginal. Moreover, the credible bands are slightly wider, reflecting increased uncertainty with more missingness.

\begin{figure}[H]
    \centering
    \includegraphics[page=1, width=.47\textwidth]{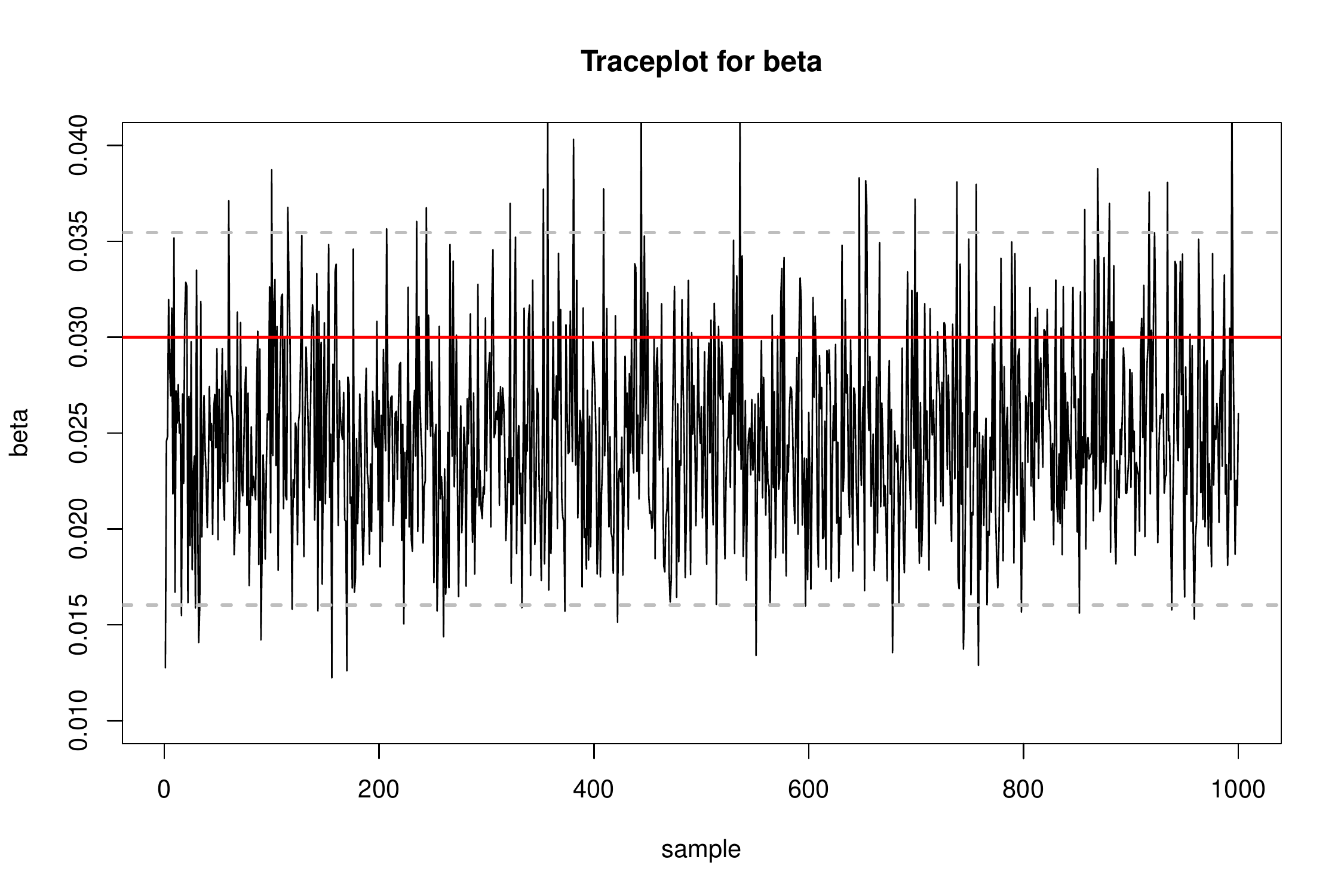}
    \includegraphics[page=2, width=.47\textwidth]{figures/coupled_6_50_try3.pdf}
    \includegraphics[page=3, width=.47\textwidth]{figures/coupled_6_50_try3.pdf}
    \includegraphics[page=6, width=.47\textwidth]{figures/coupled_6_50_try3.pdf}
    \caption{Inference results for parameters $\beta,\gamma,\alpha_{SS},\omega_{SS}$ with 50\% recovery time missingness. The uncertainty in exact recovery times does affect the estimation of the type-dependent edge rates, but not detrimentally (all the true parameter values fall into the 95\% credible intervals of the posteriors).}
    \label{fig:Bayes-miss6-50}
\end{figure}

% changed to a unified y-scale version
\begin{figure}[H]
    \centering
   \includegraphics[page=1, width=.47\textwidth]{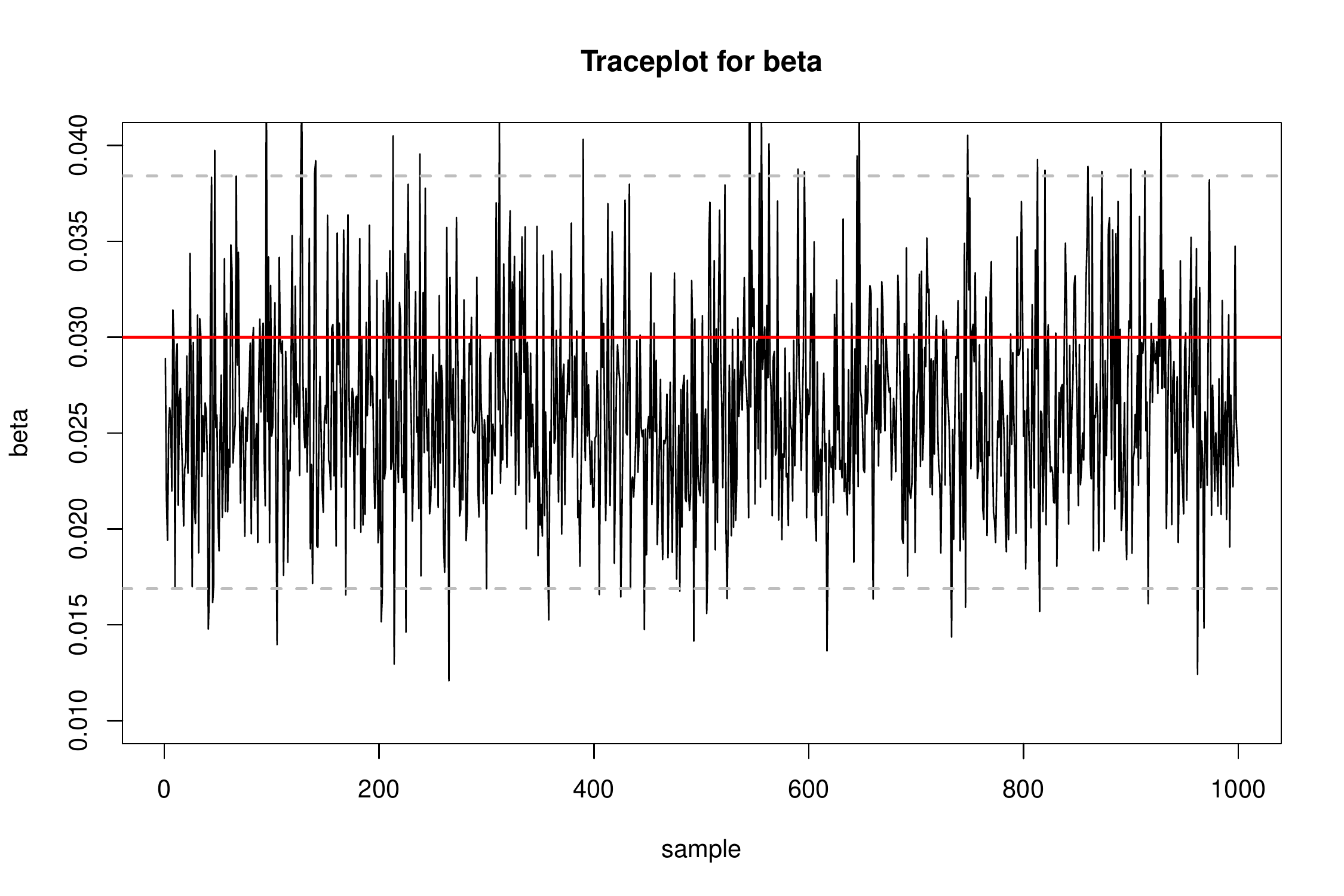}
    \includegraphics[page=2, width=.47\textwidth]{figures/coupled_6_100_try3.pdf}
    \includegraphics[page=3, width=.47\textwidth]{figures/coupled_6_100_try3.pdf}
    \includegraphics[page=6, width=.47\textwidth]{figures/coupled_6_100_try3.pdf}
    \caption{Inference results for parameters $\beta,\gamma,\alpha_{SS},\omega_{SS}$ with complete (100\%) missingness in \added{exact} recovery times. \added{The y-axis scale for each plot is the same as that in Figure~\ref{fig:Bayes-miss6-50} for easier comparison.} %For each parameter, 1000 consecutive MCMC samples are kept (dark solid line) after a 200-iteration burn-in period. 95\% sample credible intervals are bounded by the gray dashed lines, and the true parameter values are marked by red horizontal lines. 
    \added{Wider credible intervals are seen for $\beta$ and $\gamma$, but still, }all the true parameter values fall into the 95\% credible intervals of the posteriors, suggesting the capability of the inference algorithm to estimate parameters even when there is uncertainty in individual epidemic histories.}
    \label{fig:Bayes-miss6-100}
\end{figure}

\paragraph{Efficiency of the DARCI algorithm} 
We compare performance of the data augmentation algorithm stated in Proposition \ref{prop:DARCI-alg} with two  conventional sampling methods:

%\vspace*{-0.1in}
\begin{enumerate}
    \item \textbf{Rejection sampling}: Carry out Step 1 of the inference scheme via rejection sampling. For $\ell = 1:L$, keep proposing recovery times $\added{\mathbf{r}_{\ell}^*=}\{r_{\ell,i}^*\}_{i = 1:R_{\ell}} \stackrel{iid}{\sim} \text{TEXP}(\gamma^{(s-1)}, u_{\ell} , v_{\ell})$ until the proposed \added{$\mathbf{r}_{\ell}^*$} are compatible with the observed event data in $(u_{\ell} , v_{\ell}]$. We label this method by ``\textbf{Reject}''.
    \item \textbf{Metropolis-Hastings}: Modify Step 1 of the inference scheme into a Metropolis-Hastings step. For $\ell = 1:L$, propose %recovery times
    $\added{\mathbf{r}_{\ell}^*}=\{r_{\ell,i}^*\}_{i = 1:R_{\ell}} \stackrel{iid}{\sim} \text{TEXP}(\gamma^{(s-1)}, u_{\ell} , v_{\ell})$, and accept them as \added{$\mathbf{r}_{\ell}^{(s)}$} with probability
    \begin{equation*}
        \min\left(1, \frac{p\left(\mathbf{x}, \added{\mathbf{r}_{\ell}^*, \{\mathbf{r}_{\ell'}^{(s-1)}\}_{\ell' \neq \ell}} |\Theta^{(s-1)}\right) \text{pTEXP}\left(\added{\mathbf{r}_{\ell}^{(s-1)}};\gamma^{(s-1)}, u_{\ell} , v_{\ell}\right)}{p\left(\mathbf{x}, \added{\{\mathbf{r}_{\ell'}^{(s-1)}\}_{\ell'=1:L}}|\Theta^{(s-1)}\right) \text{pTEXP}\left(\added{\mathbf{r}_{\ell}^{*}};\gamma^{(s-1)}, u_{\ell} , v_{\ell}\right)} \right),
    \end{equation*}
    which equals to $1$ when the proposed \added{$\mathbf{r}_{\ell}^*$} are consistent with the observed event data in $(u_{\ell} , v_{\ell}]$ and $0$ otherwise. If the proposal is not accepted, then set \added{$\mathbf{r}_{\ell}^{(s)} = \mathbf{r}_{\ell}^{(s-1)}$}. We label this method by ``\textbf{MH}''. 
\end{enumerate}
The ``\textbf{MH}'' method employs the same principle as existing agent-based data augmentation methods \citep{cauchemez2006s,hoti2009outbreaks, fintzi2017efficient} that propose candidates of individual disease histories and accept them with probabilities computed through evaluating the likelihood (or an approximation of it) and the proposal density. In our case the implementation of ``\textbf{MH}'' is actually simpler and computationally lighter because the proposal is conditioned on known infection times and the current posterior draw of $\gamma$. The acceptance step is reduced to inspecting compatibility with known data, thus avoiding the intensive computation of likelihood evaluation. 

Although the three methods give the same inference results since they all sample from the same posterior distributions, our data augmentation algorithm (labeled by ``\textbf{DARCI}'') is more efficient than the others in two aspects. First, by drawing a new sample of recovery times from the conditional distribution in (\ref{eq:cond-recov-orig}) in each iteration, the resulting Markov chain exhibits lower autocorrelation, which leads to better mixing and fewer iterations needed to achieve a certain effective sample size. This is especially so when compared with ``\textbf{MH}''. Second, the DARCI algorithm %is generally much less time consuming when drawing from the conditional distribution 
\added{typically draws from the conditional distribution (\ref{eq:cond-recov-orig}) much more efficiently than ``\textbf{Reject}''}, because it parses out a configuration of lower bounds for imputing missing recovery times while accounting for the constraints imposed by contagion spreading and the dynamics of social links.

Three MCMC samplers are run using the three methods on the dataset showcased above. Again 1000 consecutive samples are retained for each parameter after a 200-iteration burn-in period in each case. For each parameter, we calculate the effective sample size (ESS), the Geweke Z-score \citep{geweke1991evaluating}, and the two-sided p-value for the Z-scores \added{of the resulting chains}. These results are presented in Table~\ref{tab:MCMC-diag}. Among the three methods, ``\textbf{MH}'' suffers the most from the correlation between two successive samples, while ``\textbf{DARCI}'' seems to produce high-quality MCMC samples. 

\begin{table}[ht]
\centering
%\footnotesize
\scriptsize
\caption{MCMC diagnostics for three data augmentation sampling methods, labelled as ``\textbf{DARCI}'', ``\textbf{Reject}'', and ``\textbf{MH}''. ``ESS'' stands for ``effective sample size''. The ``Z-score'' is the test statistic for MCMC convergence proposed by \cite{geweke1991evaluating}, and the two-sided p-value for each standard Z-score is also computed. Samples acquired by \textbf{MH} tend to have higher auto-correlations and thus smaller effective sample sizes.}
\begin{tabular}{lrrrrrrrrc}
  \toprule
Statistic & $\beta$ & $\gamma$ & $\alpha_{SS}$ & $\alpha_{SI}$ & $\alpha_{II}$ & $\omega_{SS}$ & $\omega_{SI}$ & $\omega_{II}$ & Method \\ 
  \midrule
ESS & 1000.00 & 1000.00 & 1000.00 & 1000.00 & 1000.00 & 1000.00 & 1000.00 & 1000.00 &  \\ 
  Z-score & -0.90 & -0.20 & -0.56 & -1.32 & -0.22 & 0.84 & -0.02 & -1.24 & \textbf{DARCI} \\ 
 $Pr(>|Z|)$ & 0.37 & 0.84 & 0.58 & 0.19 & 0.82 & 0.40 & 0.99 & 0.22 &  \\ 
  \midrule
  ESS & 1000.00 & 1160.17 & 1000.00 & 955.29 & 1000.00 & 1000.00 & 926.63 & 1000.00 &  \\ 
  Z-score & 0.48 & -1.01 & 0.44 & 0.28 & 1.08 & -2.18 & -0.16 & 0.31 & \textbf{Reject} \\ 
  $Pr(>|Z|)$ & 0.63 & 0.31 & 0.66 & 0.78 & 0.28 & 0.03 & 0.87 & 0.76 &  \\ 
  \midrule
  ESS & 566.43 & 1000.00 & 1000.00 & 1000.00 & 538.12 & 907.14 & 729.33 & 1000.00 &  \\ 
  Z-score & -1.25 & -1.83 & -0.48 & -0.59 & -2.09 & -0.24 & -1.52 & -0.57 & \textbf{MH} \\ 
  $Pr(>|Z|)$ & 0.21 & 0.07 & 0.63 & 0.55 & 0.04 & 0.81 & 0.13 & 0.57 &  \\ 
   \bottomrule
\end{tabular}
\label{tab:MCMC-diag}
\end{table}

We then compare ``\textbf{DARCI}'' and ``\textbf{Reject}'' in their running times (see Table~\ref{tab:bench-mark}). A dataset is simulated where there are different numbers of recoveries with unknown times within 5 time intervals. The two sampling methods are applied to draw a set of recovery times for each of those 5 intervals, and over multiple runs, the minimum and median times they take are recorded.  
Although the two methods draw samples from the same conditional distribution, ``\textbf{DARCI}'' tends to take less time than ``\textbf{Reject}'' in one iteration. \added{Further results suggesting scalability to larger outbreaks are available in \cref{app:sim-complete}.} 
% \jxadd{\textbf{let me know if this is convincing to mention here:} Finally, while our focus here is on efficiency, the investigation of alternate MCMC schemes illustrates how our framework can apply to analyze data with other types of missingness. When other quantities such as infection times are missing, a conditional sampler such as DARCI may not be available, but in principle one can propose latent variables analogously within an MH scheme as demonstrated above}.\alex{I'm apprehensive about having this anywhere but the discussion...}

\begin{table}[H]
    \centering
    \small
    \caption{Comparison between the two sampling methods (``\textbf{DARCI}'' and ``\textbf{Reject}'') for imputing missing recovery times. Overall, the DARCI algorithm is more efficient, especially when the number of missing recovery times is relatively large (e.g. Interval 3), or there are special constraints on viable recovery times (e.g. Interval 1, where the observed events suggest that the recovery cannot occur until half way through the time interval).}
    \begin{tabular}{cccccc}
    \toprule
    \multirow{2}{*}{Interval} & \multirow{2}{*}{\#(To recover)} & \multicolumn{2}{c}{Min Time} & \multicolumn{2}{c}{Median Time} \\
     & & \textbf{Reject} & \textbf{DARCI} & \textbf{Reject} & \textbf{DARCI}\\
    \midrule
    1 & 1 & 227$\mu$s & 224$\mu$s & 484$\mu$s & 245$\mu$s \\
    2 & 8 & 285$\mu$s & 287$\mu$s &  563$\mu$s  & 319$\mu$s \\
    3 & 15 & 163$\mu$s & 161$\mu$s  & 279$\mu$s  & 181$\mu$s \\
    4 & 2 & 138$\mu$s & 138$\mu$s & 153$\mu$s & 156$\mu$s \\
    5 & 1 & 133$\mu$s & 133$\mu$s &  146$\mu$s  & 147$\mu$s \\
    \bottomrule
    \end{tabular}
    \label{tab:bench-mark}
\end{table}

\section{Influenza-like-illnesses on a University Campus}
\label{sec:real-experiments}
In this section, we apply the proposed model and inference scheme to a real-world dataset on the transmission of influenza-like illnesses among students on a university campus.

\subsection{Data Overview}
The data we analyze in this section were collected in a 10-week network-based epidemiological study, eX-FLU \citep{aiello2016design}. The study was originally designed to investigate the effect of social intervention on respiratory infection transmission. 590 university students enrolled in the study and were asked to respond to weekly surveys on influenza-like illness (ILI) symptoms and social interactions. 103 individuals further participated in a sub-study in which each study subject was provided a smartphone equipped with an application, iEpi. The application pairs smartphones with other nearby study devices via Bluetooth, recording individual-level social interactions at five-minute intervals.

The sub-study using iEpi was carried out from January 28, 2013 to April 15, 2013 (from week 2 until after week 10). Between weeks 6 and 7, there was a one-week spring break (March 1 to March 7), during which the volume of recorded social contacts dropped noticeably. In our experiments, we use data collected on the $N=103$ sub-study participants from January 28 to April 4 (week 2 to week 10), and treat the two periods before and after the spring break as two separate and independent observation periods ($T_{\text{max}} = 31$ days for period 1 and $T_{\text{max}} = 28$ days for period 2). 

Summary statistics of the data are provided in Table~\ref{tab:realdata-overview}. Overall, infection instance counts peaked in the middle of each observation period and dropped at the end, and the dynamic social network was quite sparse; more activity (in both the epidemic process and network process) was observed in the weeks before the spring break.
Further details on data cleaning and pre-processing are provided in \cref{app: real-data}.

\begin{table}[ht]
    \centering
    \caption{Summary statistics of the real data (processed) by week: number of new infection cases (top row), maximum network density (middle row), and minimum network density (bottom row). No new infection cases took place in week 2, but two participants were already ill at the beginning of the week. The dynamic network remained sparse throughout the duration of the sub-study, except for one instance in week 3---the unusually high network density only occurred on the night of February 4, possibly due to a large-scale on-campus social event.}
    \small
    \begin{tabular}{lccccc}
    \toprule
    Week  & Wk 2 & Wk 3 & Wk 4 & Wk 5 & Wk 6  \\
    \hline
    \#(Infections) &  0  & 3 & 5 & 4 & 2 \\
    Max. Density & 0.0053 & 0.2048 & 0.0040 & 0.0038 & 0.0044 \\
    Min. Density & 0.0000 & 0.0000 & 0.0000 & 0.0002 & 0.0000\\
    %\bottomrule
    \end{tabular}
    \begin{tabular}{lccccc}
    \toprule
    Week  & (break) & Wk 7 & Wk 8 & Wk 9 & Wk 10  \\
    \hline
    \#(Infections) &  N.A.  & 1 & 3 & 5 & 1 \\
    Max. Density & N.A. & 0.0032 & 0.0032 & 0.0032 & 0.0023  \\
    Min. Density & N.A. & 0.0000 & 0.0000 & 0.0000 & 0.0000\\
    \bottomrule
    \end{tabular}
    \label{tab:realdata-overview}
\end{table}

\subsection{Analysis}
\label{sec:realdata-res}
Since the 103 individuals are sub-sampled from the 590 study participants, which are also sub-sampled from the university campus population, we treat the real data as observed on an open population. Following the parametrization introduced at the end of Section~\ref{sec:likelihood-and-inference}, we include the parameter $\xi$ to denote the rate of infection from an external source for each susceptible individual. Every infectious individual that came into contact with any infectives within 3 days prior to the onset of symptoms is regarded as an internal case (governed by parameter $\beta$); otherwise the infection is labeled as an external case (governed by parameter $\xi$).
This enables the inference procedure stated at the end of Section~\ref{sec:missrecov}.%, while making minor adaptations to accommodate the existence of external infection cases. Details on the modified inference scheme are included in Appendix \ref{app:subpopulation-inference}

%To apply our model to the data, slight modifications have to be made. Since the 103 individuals are sub-sampled from the 590 study participants, which are also sub-sampled from the entire population on the university campus, the ``closed population'' assumption implied by the generative model in Sec.~\ref{sec:model-overview} is no longer valid. This directly affects the mechanism of the infection process, where an individual can contract the disease either from another member of the sub-population (an ``internal'' infection) or from an outsider (an ``external'' infection). Therefore, we introduce an additional parameter, $\xi$, to denote the rate of getting infected from an external source for each susceptible individual, and retain the parameter $\beta$ as the rate of disease transmission per S-I link inside the sub-population.

%For every infection event in the data, if the individual came into contact with any infectives within 3 days prior to the onset of symptoms, then we regard it as an internal case, otherwise it is labeled as an external case. Then the inference procedure stated in Sec.~\ref{sec:missrecov} is employed, with minor adaptations to accommodate the existence of external infection cases. Details on the modified inference scheme are included in Appendix \ref{app:subpopulation-inference}.

The data collected during the two observation periods are considered as  independent realizations of the same adaptive network epidemic processes.
%We run the inference algorithm separately on the data collected in the two observation periods. 
For each parameter, 
%in $\Theta = \{\beta, \gamma, \alpha_{SS},\alpha_{SI},\alpha_{II}, \omega_{SS},\omega_{SI},\omega_{II}\}$ and for parameter $\xi$, 
samples drawn in the first 500 iterations are discarded and then every other sample is retained in the next 2000 iterations, resulting in 1000 posterior samples. 
%Table~\ref{tab:realdata-res-Bayes} lists the means and standard deviations of the posterior samples for a selection of parameters estimated from each of the observation periods. 
Table~\ref{tab:realdata-res-Bayes-joint} summarizes the posterior sample means and the lower and upper bounds of 95\% sample credible intervals for a selection of parameters. The output from one chain is presented here; repeated runs (with different initial conditions and random seeds) yield similar results.

% \paragraph{\added{Sensitivity Analysis}}

\added{
The data provide symptom onset times for flu-like illnesses, which can serve as proxies for the actual infection times, but the former are on average 2 days later than the latter} \citep{CDC-flu-facts}.
\added{To address this issue, we assume that the real infection times may be somewhere between 0 and 3 days prior to symptom onset (see Supplement S5.1) and thus randomly sample the latent infection times to generate multiple data versions for inference.  
% Specifically, for each symptomatic individual $i$ with recorded symptom onset time $t_i$, we randomly draw an infection time uniformly on the interval $[t_i - 0, t_i - 3]$, and use that in inference. 
% We assess sensitivity to this procedure by repeating the imputation $10$ times. The ``Multi-SD'' column of Table~\ref{tab:realdata-res-Bayes-joint} summarizes the differences in inference among imputed datasets.
% To assess sensitivity, this imputation procedure is repeated $10$ times to create multiple versions of the data; the standard deviation of parameters posterior sample means across multiple adjustments are also included in Table~\ref{tab:realdata-res-Bayes-joint} (``Multi-SD'' column). 
Results are similar across different versions of data, suggesting that inference is robust to this choice of handling possible latency periods (this is summarized by the ``Multi-SD'' column of Table~\ref{tab:realdata-res-Bayes-joint}).
}

\begin{table}[ht]
    \centering
    \caption{\added{Posterior sample means and 95\% credible intervals of select parameters (first 3 columns) obtained by the Bayesian inference scheme modified from that in Section~\ref{sec:missrecov}. Inference is carried out jointly on the two periods before and after the spring break. The final column summarizes the standard deviations of posterior sample means across $10$ versions of data, in which
    infection times are sampled randomly (and randomly) between $0$ to $3$ days prior to symptom onset.}}
    \begin{tabular}{lrrr|r}
    \toprule
    \multirow{2}{*}{Parameter} & Posterior & 2.5\% & 97.5\% & \multirow{2}{*}{Multi-SD} \\
     &  Mean & Quantile & Quantile &  \\
    \hline %\hline
    $\beta$ (internal infection) &  $0.0695$ & $0.0247$ & $0.1500$ & $0.0074$\\
    $\xi$ (external infection) & $0.00331$ & $0.00208$ & $0.00494$ & $1.797 \times 10^{-5}$\\
    $\gamma$ (recovery) & $0.294$ & $0.186$ & $0.428$ &  $0.0108$\\
    $\alpha_{SS}$ ($S$-$S$ link activation) & $0.0514$ & $0.0499$ & $0.0529$ & $0.0002$ \\
    $\omega_{SS}$ ($S$-$S$ link termination) & $38.26$ & $33.55$ & $40.62$ & $0.2522$ \\
    $\alpha_{SI}$ ($S$-$I$ link activation) & $0.130$ & $0.0785$ & $0.194$ & $0.0097$ \\
    $\omega_{SI}$ ($S$-$I$ link termination) & $53.5$ & $22.5$ & $231.7$ & $31.4092$\\
    \bottomrule
    \end{tabular}
    \label{tab:realdata-res-Bayes-joint}
\end{table}

Our findings suggest that flu-like symptoms spread quite slowly but recoveries are made rather quickly.
\added{For instance, if a susceptible person maintains \textbf{one} infectious contact, then he has a probability of approximately $6.71\%$ to contract infection through such contact \textbf{within one day}, and yet it takes (on average) a little more than 3 days for someone to no longer feel ill after infection. }
% {On average, it takes about 14 days (2 weeks) of contact with one infectious person inside the population for a susceptible individual to start showing symptoms, yet it takes a little more than 3 days for someone to no longer feel ill. }
The external infection force is non-negligible: given the number of susceptibles in the population (typically about 100), the population-wide external infection rate is approximately $100 \times 0.0033 = 0.33$, implying that an external ILI case is expected to occur every other three days. This is a reasonable estimate consistent with having observed 9 external infection cases within $28$ days during the second period.

The inferred link rates reflect an interesting pattern in social interactions in this particular population: individuals are reluctant to establish contact and active contacts are broken off quickly---an average pair of healthy people initiate/restart their interaction after waiting 20 days and then end it after spending less than 40 minutes together. Moreover, it seems that on average a healthy-ill link is activated more frequently than a healthy-healthy link, but the former is also terminated faster---this might be because those students who fell ill in the duration of the study happen to be more socially proactive, but once their healthy social contacts realize they are sick and thus potentially infectious, the contact is cut short to avoid disease contraction. 

It is also notable that the sample 95\% credible intervals for $\beta$, $\alpha_{SI}$ and $\omega_{SI}$ are relatively wide, indicating a high level of uncertainty in the estimation for these parameters. It is challenging to estimate the internal infection rate $\beta$  because dataset contains only 6 cases of internal infection in total (5 in period 1, 1 in period 2), providing limited information on the rate of transmission. Similar issues are present for the estimation of $\alpha_{SI}$ and $\omega_{SI}$; since there were no more than 5 infectious individuals at any given time, network events related to them were few and far between. Moreover, since their exact recovery times are unknown, there is additional uncertainty associated with their exact disease statuses when they activated or terminated social links. Such measure of uncertainty, readily available through stochastic modeling and Bayesian inference, provides valuable insights into the amount of information the data contain and the level of confidence we possess when making conclusions and interpretations. The inference outcomes imply that, for example, the real data sufficiently inform the contact patterns among healthy individuals in this population but are limited toward understanding how long a healthy person and a symptomatic person typically maintain their contact.

\section{Discussion}
\label{sec:discussion}

%In this paper we formulated a stochastic model to describe the interplay of SIR-type epidemic processes on temporal networks and the adaptive dynamics of networks. Upon deriving the complete data likelihood, we developed an inference method equipped with a data augmentation algorithm that efficiently deals with partial epidemic observations. The model and the inference framework are assessed in experiments with simulated data and a real-world dataset. 

This paper has focused on enabling inference for partially observed epidemic processes on dynamic and adaptive networks. We formulated a continuous-time Markov process model to describe the epidemic-network interplay and derived its complete data likelihood. This leads to the design of conditional sampling techniques that enable data augmented inference methods to accommodate missingness in individual recovery times. 

%\subsection{Extensions}
There are several limitations and natural extensions of our model. First, we address the issue of a latency period here by using a sensitivity analysis of the symptom onset time. 
% One issue that we address but do not model \textbf{explicitly} is the latency period: there is always a time lag between the contraction of disease (and beginning of infectiousness) and symptom onset. In our paper, we account for uncertainty in a possible latency period by introducing uniform random variables in addition to reported symptom onset times. 
As infectiousness is the focus of inference, we prefer this approach over, for instance, modeling an additional compartment (i.e. an SEIR model) in favor of model parsimony. We note that the latter  is possible by extending our proposed likelihood framework, but introduces additional parameters that are often hard to identify without additional data directly informative of latency or modeling assumptions regarding the latency period (e.g., non-infectious or less infectious when latent).

A second extension of the model pertains to the handling of other missing data types. Motivated by our case study in which infection-related events (symptom onset) are updated daily, whereas recoveries are only provided at a much coarser resolution (in weekly summaries), our current method focuses on imputing missing recovery times. %but not missing infection times.}

% Another possible limitation of our approach is handling missing recovery times but not missing infection times. This choice is directly motivated by the nature of our dataset: infection-related events (symptom onset) are updated \textbf{daily}, whereas recoveries are only provided in \textbf{weekly} summaries (at a much coarser resolution); 
%Similar settings are abundant in real-world epidemic data (infections are often reported daily but not recoveries in many modern ``instances data''
While it is common in real-world data to focus on new cases  \citep{WHO2003SARS, WHO2020COVID19}, one may be provided with such incidence data at coarser time resolution so that infection times must also be imputed. %However, if both infections and recoveries are observed at the same resolution (say, only weekly) or if infections are observed at a coarser scale than recoveries, 
Our framework applies in principle to such settings where missing infection times should be accounted for explicitly. One may derive analogous conditional samplers to DARCI, or at worst incur a computational tradeoff. Even without access to sampling from the exact conditional distributions, we can replace the Gibbs step to impute infection times by a Metropolis proposal within each iteration of the MCMC scheme \citep{britton2002bayesian}. 

Our contributions leverage network information to avoid a common model misspecification, but the current methods are limited to scenarios where such information is completely informed. %, due to the nature of our example dataset (and that similar datasets are burgeoning with the increased use of mobile devices in recent studies). 
If network dynamics are only partially observed---for instance link events are missing or exist only on a weekly survey basis---our proposed methods do not immediately apply, but can be extended via further data augmentation over unobserved network event times. Doing so falls under the same likelihood-based framework, yet practical challenges related to mixing of the Markov chain may arise due to the increased latent space.  Another viable strategy is to adopt a discrete time model for network evolution that can be seen either as an alternative or an approximation to the link-Markovian jump process we propose. These directions remain open for future work. % we propose by using a discrete Markov chain to describe weekly changes to the ego networks. 

%\subsection{Reinterpreting $R_0$}
We have demonstrated that accounting for changes in the contact structure is critical to accurately estimate disease parameters
such as infection and recovery rates. %$\beta$ and $\gamma$. 
Because our model is 
% can be seen as 
a generalization of existing compartmental models, such rates are consistent with their definitions and interpretation in existing literature. Other quantities such as the basic reproductive number $R_0$, defined as the average number of infections caused by an infectious individual, do not translate as readily \citep{tunc2013epidemics,van2010adaptive}. This is the case when disease and network properties are conflated: the basic reproductive number depends on the product of infection rate $\beta$ and number of contacts. The effect of interventions such as quarantine are often incorporated similarly, for instance by way of a change-point in $\beta$ \citep{ho2018direct}, yet such policies should naturally translate to changes in the contact network rather than the inherent infectivity of the disease. Because our model mechanistically describes the joint dynamics of the network and the disease spread, such phenomena can now be modeled explicitly in terms of network parameters rather than indirectly through disease parameters, leading to more accurate and  interpretable inference.
The proposed framework thus serves as a point of departure to further explore these promising extensions.

\bigskip
\begin{center}
{\large\bf SUPPLEMENTARY MATERIAL}
\end{center}

\begin{description}

\item[Supplementary information:] Supplementary proofs and derivations, inference details on open population epidemics, and more results from simulation experiments and real data experiments. (.pdf file: \texttt{supplement.pdf})

\item[Codes and examples:] R codes for all simulation experiments, accompanied by example synthetic datasets. (Anonymized repository: \url{https://anonymous.4open.science/r/2231b6ae-00aa-414c-9d6f-37c69084e5a0/}) 

\end{description}

%\newpage
%\bibliography{ref}

% comment out to put everything back
%\processdelayedfloats 
%\clearpage

\newpage
\appendix
%\appendixpage

%\begin{appendices}
\begin{center}
    {\Large\bf SUPPLEMENTARY INFORMATION}
\end{center}
\renewcommand{\thesection}{S\arabic{section}}

% comment out the "post..." lines to put everything back in their place
\renewcommand{\thetable}{S\arabic{table}}%
\renewcommand{\thefigure}{S\arabic{figure}}%

\setcounter{table}{0}
\setcounter{figure}{0}
%\setcounter{postfigure}{0}

%\crefalias{section}{Supplement}
\pdfoutput=1

\section{Complete Data Likelihood for SIS-type Contagions}
\label{app:sis-lik}
Here we consider an SIS-type infectious disease to illustrate how our methodology can be applied to similar model classes.  The complete data likelihood can be derived following the same steps in Section~\ref{sec:likelihood-and-inference}. Alternatively, one can slightly modify (\ref{eq:SIR-comp-lik}) to arrive at the complete likelihood for an SIS-type contagion. Since an individual doesn't acquire immunity upon recovery, it is equivalent to setting $H(t)\equiv S(t)$ at any time $t$. Thus the complete data likelihood is
\begin{align} \label{eq:SIS-comp-lik}
    &\mathcal{L}(\beta, \gamma, \tilde\alpha, \tilde\omega|\mathcal{G}_0) \nonumber \\
    =& \gamma^{n_R} \beta^{n_E - 1} \alpha_{SS}^{C_{SS}}\alpha_{SI}^{C_{SI}}\alpha_{II}^{C_{II}}\omega_{SS}^{D_{SS}}\omega_{SI}^{D_{SI}}\omega_{II}^{D_{II}} \prod_{j=2}^n \left[ \tilde{M}(t_j) \left(I_{p_{j1}}(t_j)\right)^{F_j} \right] \nonumber \\
    & \times \exp\left(-\int_0^{T_{\max}}\left[\beta S\kern-0.14em I(t) + \gamma I(t) + \tilde{\alpha}^T \mathbf{M}_{\max}(t) + (\tilde{\omega}-\tilde{\alpha})^T \mathbf{M}(t) \right]dt\right).
\end{align}

\section{Auxiliary Proofs and Derivations}
\label{app:aux-proofs}
\paragraph{Proof for Theorem \ref{thm:MLEs}} 
From (\ref{eq:SIR-comp-lik}), we can obtain the log-likelihood:
\begin{align} \label{eq:SIR-log-lik}
    &\ell(\beta, \gamma, \tilde\alpha, \tilde\omega|\mathcal{G}_0) =\log \mathcal{L}(\beta, \gamma, \tilde\alpha, \tilde\omega|\mathcal{G}_0)  \nonumber \\
    =&  \sum_{j=2}^n \left[\log \tilde{M}(t_j) + F_j \log \left(I_{p_{j1}}(t_j)\right)\right] + n_R \log \gamma + (n_E - 1) \log\beta  \\ 
    & + C_{HH} \log \alpha_{SS} + C_{HI}\log  \alpha_{SI} + C_{II} \log \alpha_{II} + D_{HH} \log\omega_{SS} + D_{HI}\log\omega_{SI} +D_{II}\log \omega_{II} \nonumber \\
    & -\sum_{j=1}^{n} \left[\beta SI(t_j) + \gamma I(t_j) + \tilde{\alpha}^T(\mathbf{M}_{\max}(t_j)-\mathbf{M}(t_j)) + \tilde{\omega}^T\mathbf{M}(t_j) \right] (t_j - t_{j-1}). \nonumber 
\end{align}
Taking partial derivatives of the right hand side of (\ref{eq:SIR-log-lik}) with respect to the parameters and setting them to zero yield the results above.

\section{Relaxing the Closed Population Assumption}
\label{app:subpopulation-inference}
Suppose the observed population is not fully closed, but is a subset of a larger yet unobserved population. Then it is possible for an individual to get infected by an outsider. Let $\xi$ be the ``external infection'' rate, the rate for any susceptible individual to be infected by any external infectious source, then the complete data likelihood is
\begin{align} 
\label{eq:SIR-lik-external-1}
    &\mathcal{L}(\beta, \xi, \gamma, \tilde\alpha, \tilde\omega|\mathcal{G}_0) = p(\text{epidemic events},\text{network events}|\beta, \xi, \gamma, \alpha, \omega,\mathcal{G}_0) \nonumber \\
    =& \gamma^{n_R} %\beta^{n_E - 1}
    \alpha_{SS}^{C_{HH}}\alpha_{SI}^{C_{HI}}\alpha_{II}^{C_{II}}\omega_{SS}^{D_{HH}}\omega_{SI}^{D_{HI}}\omega_{II}^{D_{II}} \prod_{j=2}^n \left[ \tilde{M}(t_j) \left(\beta I_{p_{j1}}(t_j) + \xi\right)^{F_j} \right] \nonumber \\
    & \times \exp\left(-\int_0^{T_{\max}}\left[\beta S\kern-0.14em I(t) + \xi S(t) + \gamma I(t) + \tilde{\alpha}^T\mathbf{M}_{\max}(t) + (\tilde{\omega}-\tilde{\alpha})^T\mathbf{M}(t) \right]dt\right).
\end{align}

MLEs for $\{\gamma, \tilde{\alpha}, \tilde{\omega}\}$ remain unchanged, but estimating $\beta$ and $\xi$ is less straightforward. Let $\ell(\beta, \xi, \gamma, \tilde\alpha, \tilde\omega|\mathcal{G}_0)$ be the log-likelihood, then the partial derivatives of the log-likelihood w.r.t. $\beta$ and $\xi$ are
\begin{align*}
    \frac{\partial \ell}{\partial \beta} &= \sum_{j=2}^n \frac{F_j I_{p_{j1}}(t_j)}{\beta I_{p_{j1}}(t_j) + \xi} - \sum_{j=1}^n SI(t_j)(t_j - t_{j-1}),\\
    \frac{\partial \ell}{\partial \xi} &= \sum_{j=2}^n \frac{F_j}{\beta I_{p_{j1}}(t_j) + \xi} - \sum_{j=1}^n S(t_j)(t_j - t_{j-1}),
\end{align*}
which do not directly lead to closed-form solutions.

Reparameterizing by $\xi = \kappa \beta$ leads to the following partially derivatives
\begin{align}
    \frac{\partial \ell}{\partial \beta} &= \frac{n_E - 1}{\beta} -  \sum_{j=1}^n [SI(t_j)+\kappa S(t_j)](t_j - t_{j-1}), \label{eq:beta-partial-2}\\
    \frac{\partial \ell}{\partial \kappa} &= \sum_{j=2}^n \frac{F_j}{ I_{p_{j1}}(t_j) + \kappa} - \beta \sum_{j=1}^n S(t_j)(t_j - t_{j-1}), \label{eq:kappa-partial}
\end{align}
which are slightly more straightforward in form, and can be solved numerically to obtain the MLEs.

% If $\xi = \kappa \beta$ with the ratio $\kappa$ known (or somehow estimated beforehand), then it is simple to obtain the MLE for $\beta$:
% \begin{equation*}
%     \hat \beta = \frac{n_E - 1}{\sum_{j=1}^n[SI(t_j) + \kappa S(t_j)](t_j - t_{j-1})}.
% \end{equation*}

If, somehow, we have information on which infection cases are caused by internal sources and which are caused by external sources, then we can directly obtain the MLEs and Bayesian posterior distributions for all the parameters. For an infection event $e_j$ (with $F_j=1$), let $\text{Int}_j = 1$ if it is ``internal'' and let $\text{Int}_j = 0$ otherwise. Then the complete data likelihood can be re-written as
\begin{align} 
\label{eq:SIR-lik-external-2}
    &\mathcal{L}(\beta, \xi, \gamma, \tilde\alpha, \tilde\omega|\mathcal{G}_0)\nonumber \\
    =& \beta^{\left(n_E^{\text{int}}-\text{Int}_1\right)} \xi^{\left(n_E^{\text{ext}}-1+\text{Int}_1)\right)} \gamma^{n_R}
    \alpha_{SS}^{C_{HH}}\alpha_{SI}^{C_{HI}}\alpha_{II}^{C_{II}}\omega_{SS}^{D_{HH}}\omega_{SI}^{D_{HI}}\omega_{II}^{D_{II}} \prod_{j=2}^n \left[ \tilde{M}(t_j) I_{p_{j1}}(t_j)^{F_j\text{Int}_j} \right] \nonumber \\
    & \times \exp\left(-\int_0^{T_{\max}}\left[\beta S\kern-0.14em I(t) + \xi S(t) + \gamma I(t) + \tilde{\alpha}^T\mathbf{M}_{\max}(t) + (\tilde{\omega}-\tilde{\alpha})^T\mathbf{M}(t) \right]dt\right),
\end{align}
where $n_E^{\text{int}}$ and $n_E^{\text{ext}}$ are the total numbers of internal and external infection events, respectively. 

Estimation for all parameters remains unchanged except for $\beta$ and $\xi$. Their MLEs are
\begin{equation}
    \hat{\beta} = \frac{n_E^{\text{int}}-\text{Int}_1}{\sum_{j=1}^{n}SI(t_j)(t_j - t_{j-1})},\quad
    \hat{\xi} = \frac{n_E^{\text{ext}}-1+\text{Int}_1}{\sum_{j=1}^{n}S(t_j)(t_j - t_{j-1})},
\end{equation}
and with Gamma priors $\beta \sim Ga(a_{\beta},b_{\beta})$ and $\xi \sim Ga(a_{\xi},b_{\xi})$, their posterior distributions are
\begin{align}
    \label{eq:beta-int-posterior}
    \beta | \{e_j\} &\sim Ga\left(a_{\beta} + (n_E^{\text{int}}-\text{Int}_1), b_{\beta} + \sum_{j=1}^{n}SI(t_j)(t_j - t_{j-1})\right),\\
    \label{eq:xi-ext-posterior}
    \xi | \{e_j\} &\sim Ga\left(a_{\xi} + (n_E^{\text{ext}}-1+\text{Int}_1), b_{\xi}+ \sum_{j=1}^{n}S(t_j)(t_j - t_{j-1})\right).
\end{align}

When there is missingness in recovery times, the Bayesian inference procedure described in Section~\ref{sec:missrecov} can still be carried out, with two slight modifications. First, in the data augmentation step, when drawing missing recovery times in an interval $(u, v]$, the DARCI algorithm ( Prop.~\ref{prop:DARCI-alg}) inspects $\mathcal{I}_p$ only for each $p \in \mathcal{P}^{\text{int}}$, where $\mathcal{P}^{\text{int}}$ is the group of individuals who get \emph{internally} infected during $(u, v]$. Second, in each iteration, parameter values are drawn from the posterior distributions specified in (\ref{eq:posteriors}) except for $\beta$ and $\xi$, for which the posteriors are stated in (\ref{eq:beta-int-posterior}) and (\ref{eq:xi-ext-posterior}), respectively.

\section{Flexible Network Dynamics under the Generative Model}
\label{app:network-behavior}

As stated in the main text, the proposed generative model can accommodate \textbf{arbitrary} initial network structures. Moreover, due to the \textbf{coupled} nature of the epidemic process and the network process, different choices of parameters can lead to a wide variety of network behaviors. 

Here we demonstrate via simulations that, even if the initial network follows a simple, idealized model, it can still evolve and adapt to exhibit characteristics that are no longer simple and idealistic throughout the process. 

Setting the population size to $N=500$, we assume two different models for the initial network: 1) a random Erdős–Rényi graph (``ER''), and 2) a random \textbf{scale-free} network (sampled using the Barabasi-Albert model (``BA'') \citep{barabasi1999emergence}). The ``ER'' model leads to a degree distribution close to Poisson, whereas the ``BA'' model leads to power-law degree distribution.

Using the following parameter settings (same as those in most of our simulations),
\begin{align*}
\label{eq:settings}
    &\beta = 0.03, \gamma = 0.12;\\ 
    &\tilde{\alpha}^T = (\alpha_{SS},\alpha_{SI},\alpha_{II}) =
    (0.005, 0.001, 0.005),\\ 
    &\tilde{\omega}^T = (\omega_{SS},\omega_{SI},\omega_{II}) =
    (0.05, 0.1, 0.05).
\end{align*}
we adopt a dynamic that 1) discourages new $S-I$ connections and thus effectively ``isolates'' super-spreaders probabilistically, and 2) mostly sustains the usual contact activities between $S-S$, $R-R$ and $I-I$ pairs. This, during the process, can lead to a \textbf{bi-modal} degree distribution that obviously deviates from the initial degree distribution. 

Below in Figure~\ref{fig:deg-dist-ER}, we include the empirical degree distributions at the start and the end of a typical simulation, with a random Erdős–Rényi graph as the initial network. Here the red curve represents the actual degree distribution in the simulated population, and (in contrast) the lightblue curve represents the density of a degree sequence sampled from a Poisson with the same mean degree.

For comparison, in Figure~\ref{fig:deg-dist-BA} we include similar plots with the initial network as a \textbf{scale-free} network (sampled using the Barabasi-Albert model). Note that the initial network setting here is similar to that used in \cite{volz2007susceptible}, but our dynamic network model is fundamentally different from (and more flexible) than the neighbor-exchange framework proposed in that work.

We can observe that, once the network dynamics kicks in, the initial network structure doesn't matter that much, as changes in the network links are driven jointly by the parameters and the epidemic dynamics in time.

\begin{figure}[H]
    \centering
    \includegraphics[width=0.47\textwidth,page=1]{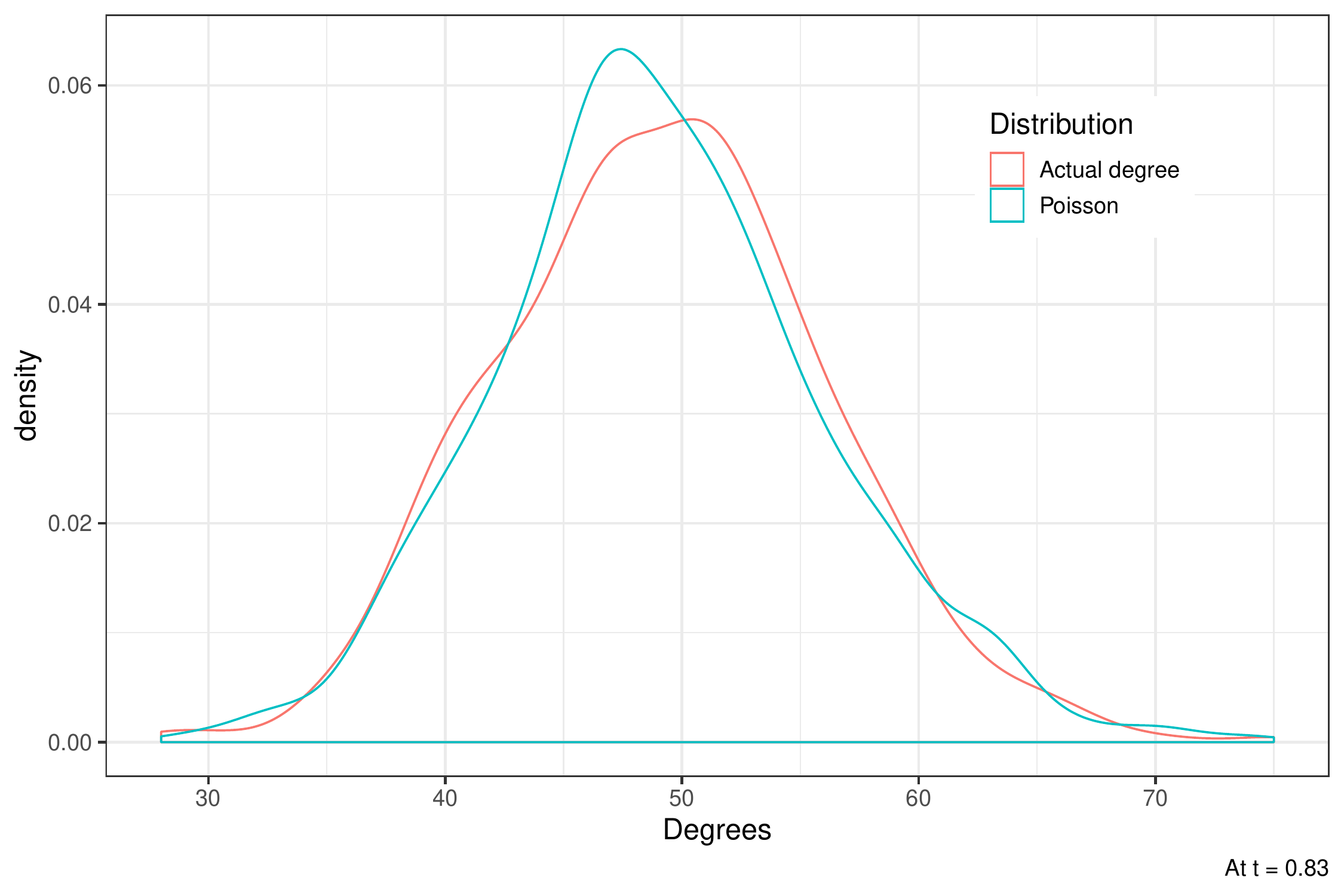}
    \includegraphics[width=0.47\textwidth,page=57]{figures/3_1.pdf}
    \caption{Actual empirical degree distribution in simulation (red) versus empirical degree distribution drawn from Poisson (blue/green). Here the initial network is a random Erdős–Rényi graph.  \textbf{Left: beginning of process; right: end of process.}}
    \label{fig:deg-dist-ER}
\end{figure}

\begin{figure}[h]
    \centering
    \includegraphics[width=0.47\textwidth,page=1]{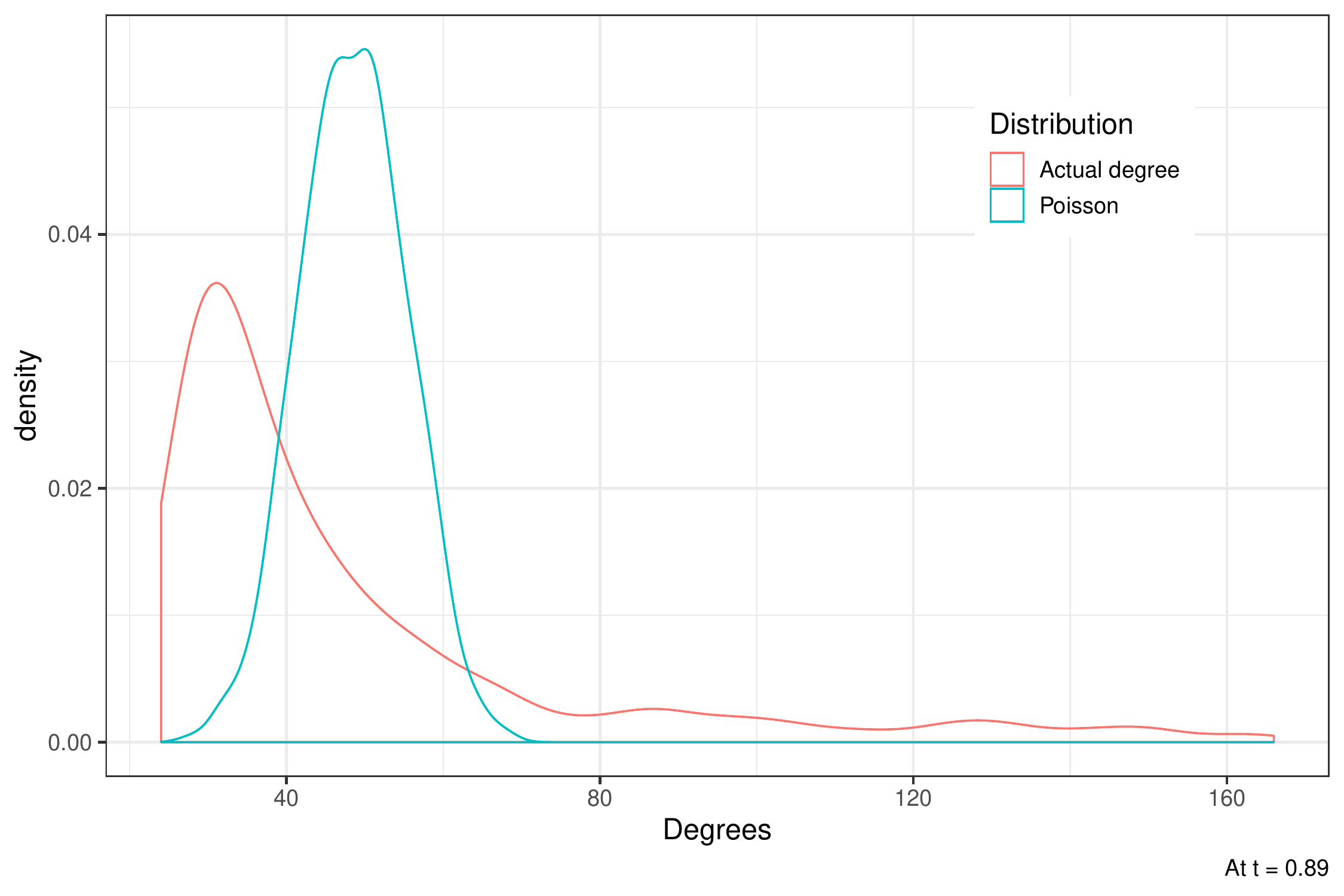}
    \includegraphics[width=0.47\textwidth,page=53]{figures/2_1.pdf}
    \caption{Actual empirical degree distribution in simulation (red) versus empirical degree distribution drawn from Poisson (blue/green). Here the initial network is a random scale-free network (sampled using the Barabasi-Albert model \citep{barabasi1999emergence}). \textbf{Left: beginning of process; right: end of process.}}
    \label{fig:deg-dist-BA}
\end{figure}

%\section{More Experimental Results}
\section{More Results on Simulation Experiments}
\label{app:sim-complete}

\paragraph{Supplement for ``inference from complete event data''}
Figure~\ref{fig:MLEs-complete-CI-app} and \ref{fig:Bayes-complete-multi-app} complement Figure~\ref{fig:MLEs-complete-CI} and \ref{fig:Bayes-complete-multi} in the main text, showing inference results for all the parameters in the corresponding experiments.

\paragraph{Experiments on larger networks}
Figure~\ref{fig:MLEs-coupled-big} shows MLEs and 95\% confidence bands for parameters with complete data generated on a network with $N=500$ individuals. Other experimental settings are the same as those in Section~\ref{sec:infer-complete}. With a larger population, there tends to be more events available for inference, so the accuracy is in fact improved.

\paragraph{Experiments on different initial network configurations} 
Still set population size $N=100$, but instead of a random Erdős–Rényi graph as $\mathcal{G}_0$, the initial network is a ``hubnet'': one individual (the ``hub'') is connected to everyone else in the population while the others form an $ER(N-1, p)$ random graph, with edge probability $p=0.1$. Figure~\ref{fig:Bayes-multi-hubnet} summarizes results of Bayesian inference carried out on complete event data generated in this setting. 

\paragraph{Supplement for ``Assessing model flexibility''} 
Estimate parameters $\Theta$ of the full model on datasets generated from 1) the decoupled temporal network epidemic process with type-independent edge rates, and 2) the static network epidemic process where the network remains unchanged. For both simpler models, fix $\beta = 0.03$ and $\gamma = 0.12$, and for the former model, let link activation rate $\alpha = 0.005$ and termination rate $\omega = 0.05$. Still, set population size $N=100$ and let the initial network be a random Erdős–Rényi graph with edge probability $p=0.1$.

We present, in Figure~\ref{fig:Bayes-multi-decoup}, the results of Bayesian inference on datasets generated from the decoupled process model. Across four different realizations, it can be observed that, the posterior samples of link activation rates ($\alpha_{SS},\alpha_{SI},\alpha_{II}$) concentrate around the same mean, and uncertainty is reduced with more events available for inference. Same can be said about the link termination rates, $\omega_{SS},\omega_{SI},\omega_{II}$. This verifies that the proposed model is indeed a generalization of the aforementioned two simpler processes, and the inference method is able to recover the truth under mild model misspecification.

%% MLE, G0 = ER(N=100) %%
\begin{figure}[H]
    \centering
    \includegraphics[width=0.95\textwidth]{figures/MLE_plots_ex3dat5_1.pdf}
    \includegraphics[width=\textwidth]{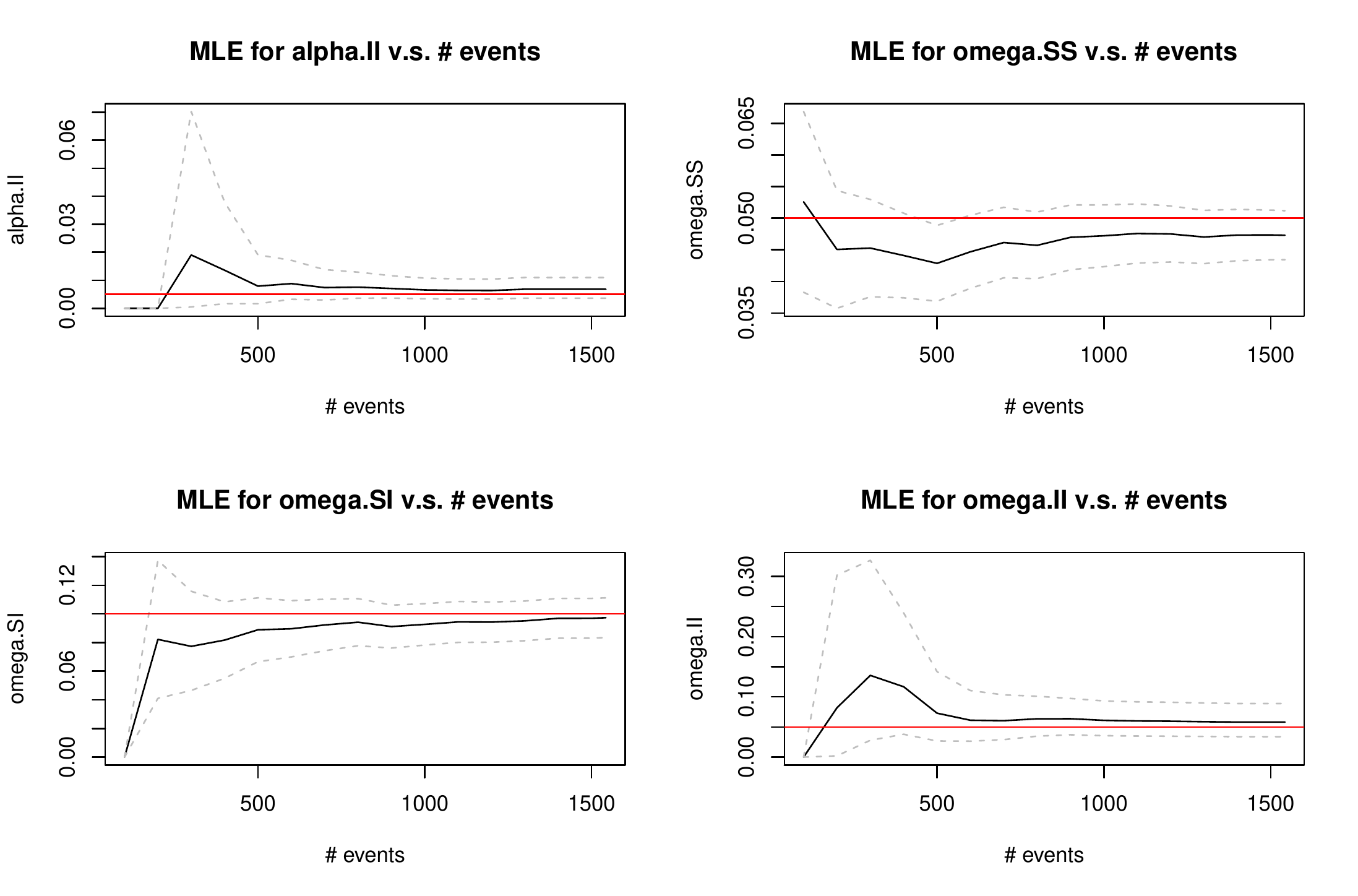}
    \caption{MLEs versus number of events used for inference. Dashed gray lines show the lower and upper bounds for 95\% frequentist confidence intervals, and red lines mark the true parameter values. %In this realization $n_E = = n_R = 48, C_{SS}=621, C_{SI}=35, C_{II}=13, D_{SS}=573, D_{SI}=189, D_{II}=17$.}
    }
    \label{fig:MLEs-complete-CI-app}
\end{figure}

%% Bayesian, G0 = ER(N=100) %%
\begin{figure}[H]
    \centering
    \includegraphics[width=0.49\textwidth,page=1]{figures/ex1_all.pdf}
    \includegraphics[width=0.49\textwidth,page=2]{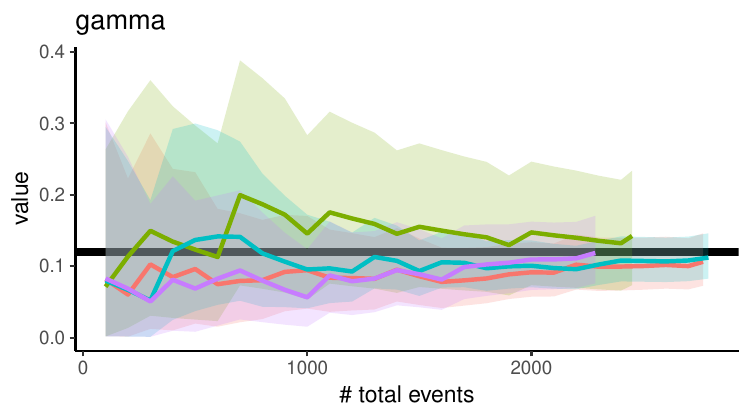}
    \includegraphics[width=0.49\textwidth,page=3]{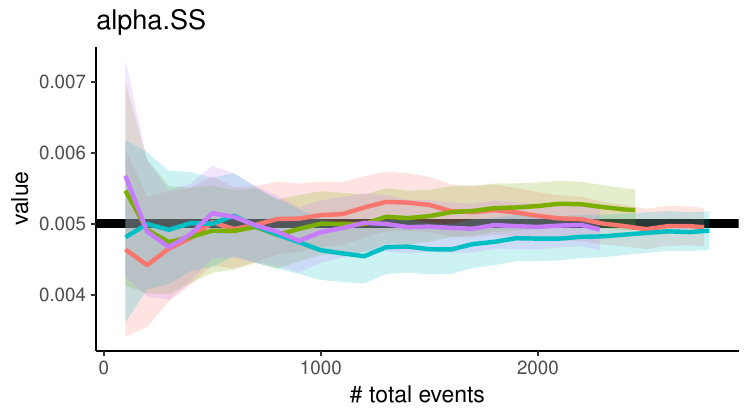}
    \includegraphics[width=0.49\textwidth,page=4]{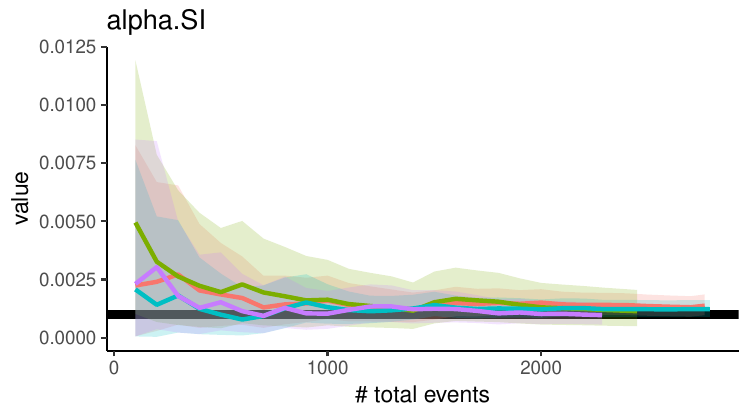}
    \includegraphics[width=0.49\textwidth,page=5]{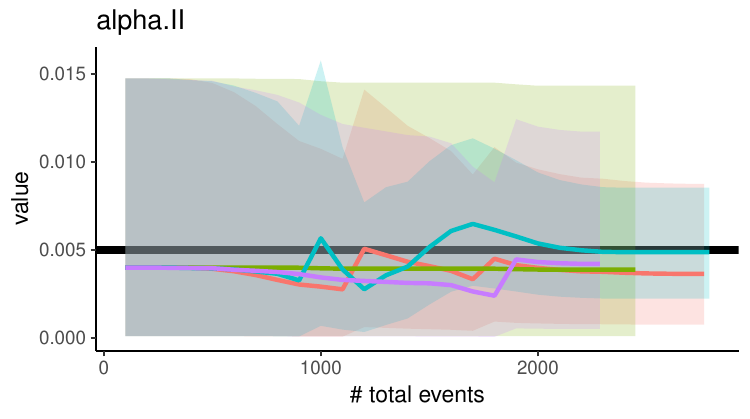}
    \includegraphics[width=0.49\textwidth,page=6]{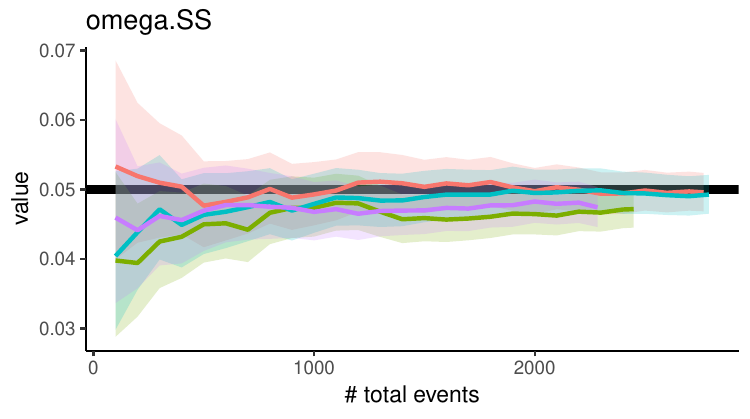}
    \includegraphics[width=0.49\textwidth,page=7]{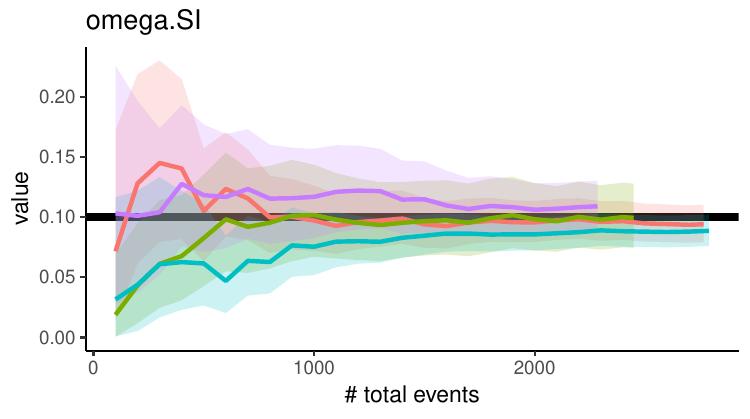}
    \includegraphics[width=0.49\textwidth,page=8]{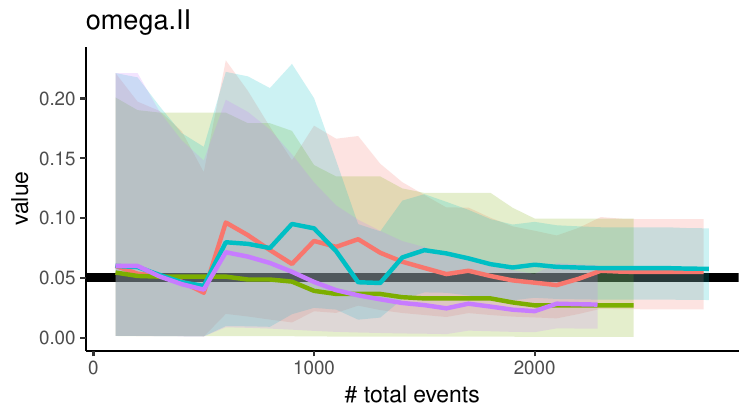}
    \caption{Posterior sample means v.s. number of total events used for inference. True parameter values are marked by \textbf{bold dark} horizontal lines, along with 95\% credible bands. Results are presented for 4 different complete datasets.}
    \label{fig:Bayes-complete-multi-app}
\end{figure}

\begin{figure}[H]
    \centering
    \includegraphics[width=0.95\textwidth,page=1]{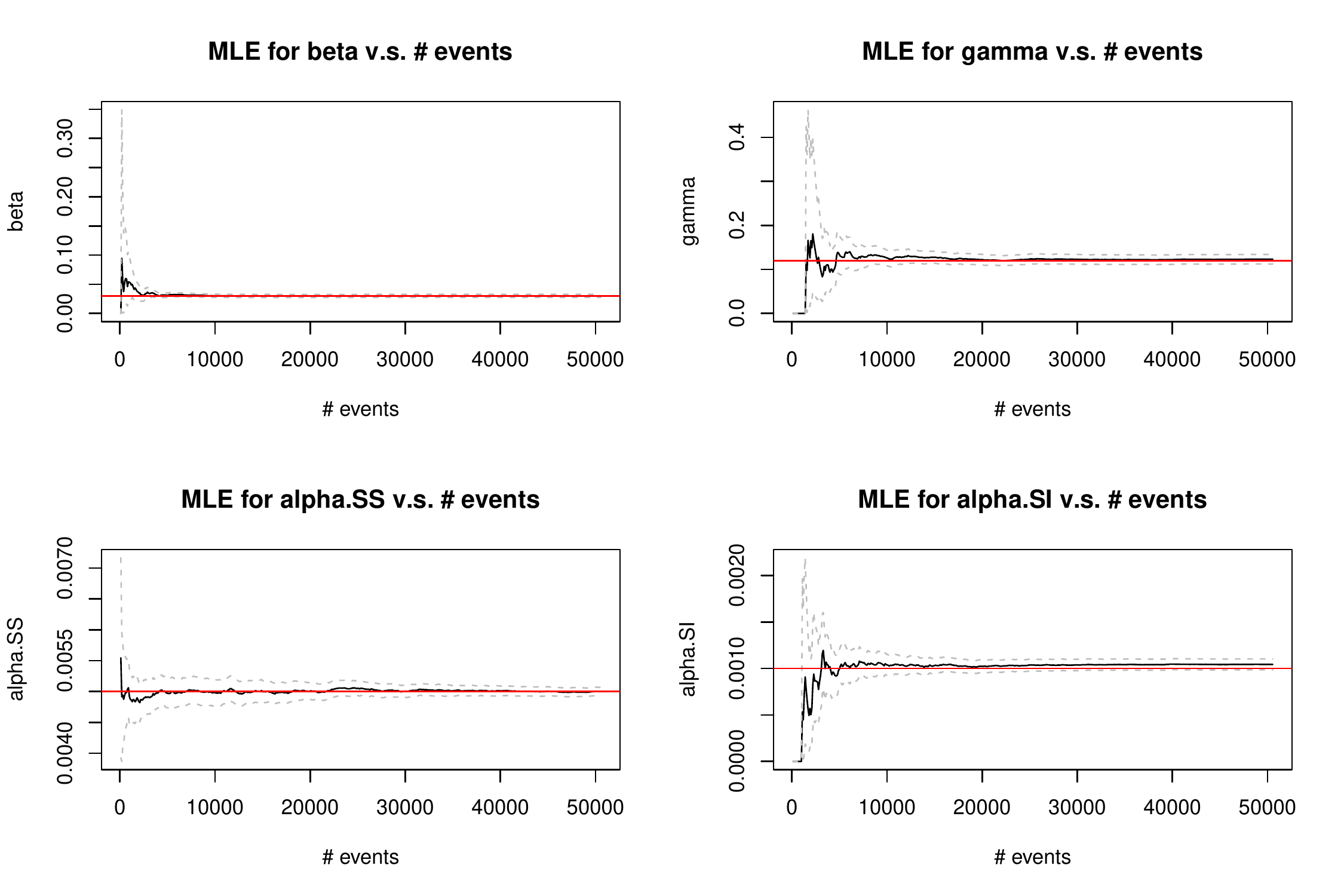}
    \includegraphics[width=0.95\textwidth,page=2]{figures/N500_MLEs.pdf}
    \caption{MLEs versus number of total events, on a larger population with $N=500$. Dashed gray lines show the lower and upper bounds for 95\% confidence intervals, and red lines mark the true parameter values. %In this realization $n_E = 497, n_R = 498, C_{SS}=20770, C_{SI}=1210, C_{II}=1745, D_{SS}=17227, D_{SI}=6855, D_{II}=1747$. 
    With a larger population size, there tends to be more events, which in fact facilitates estimation.}
    \label{fig:MLEs-coupled-big}
\end{figure}

\begin{figure}[H]
    \centering
    \includegraphics[width=0.49\textwidth,page=1]{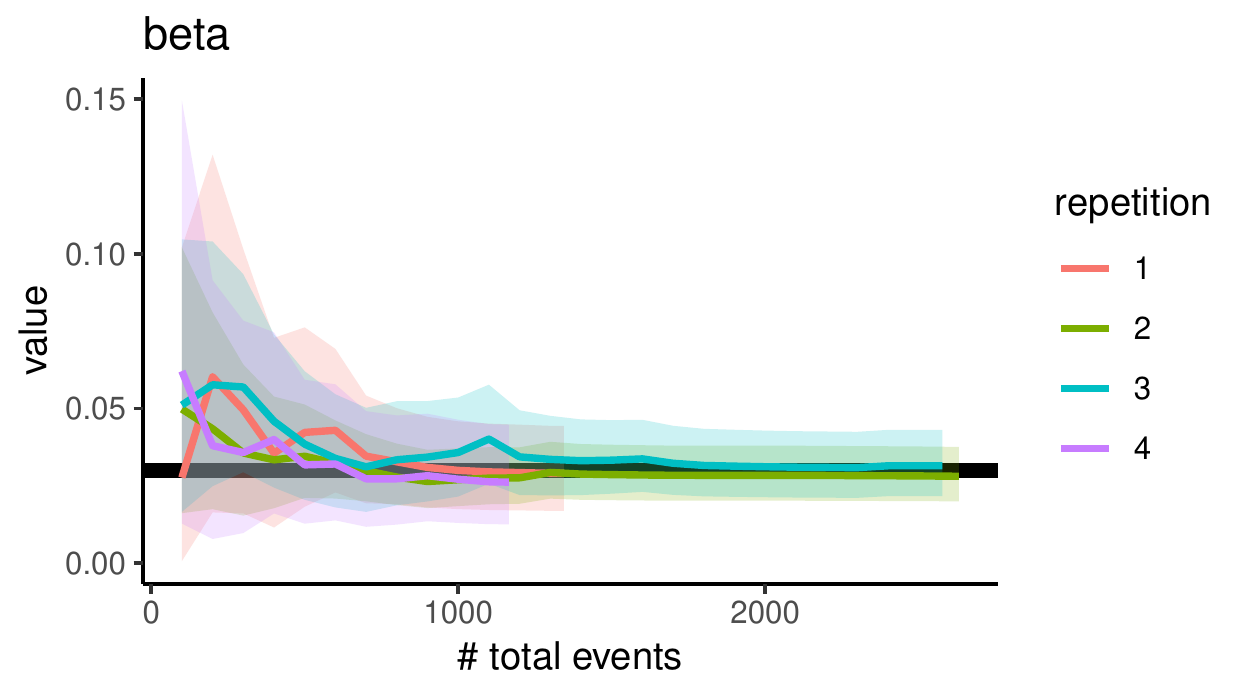}
    \includegraphics[width=0.49\textwidth,page=2]{figures/hubnet_Bayes2.pdf}
    \includegraphics[width=0.49\textwidth,page=3]{figures/hubnet_Bayes2.pdf}
    \includegraphics[width=0.49\textwidth,page=4]{figures/hubnet_Bayes2.pdf}
    \includegraphics[width=0.49\textwidth,page=5]{figures/hubnet_Bayes2.pdf}
    \includegraphics[width=0.49\textwidth,page=6]{figures/hubnet_Bayes2.pdf}
    \includegraphics[width=0.49\textwidth,page=7]{figures/hubnet_Bayes2.pdf}
    \includegraphics[width=0.49\textwidth,page=8]{figures/hubnet_Bayes2.pdf}
    \caption{Posterior sample means v.s. number of total events, with $\mathcal{G}_0$ as a $N=100$-node ``hubnet''. True parameter values are marked by \textbf{bold dark} horizontal lines, along with 95\% credible bands. Results are presented for 4 different complete datasets.}
    \label{fig:Bayes-multi-hubnet}
\end{figure}

\begin{figure}[H]
    \centering
    \includegraphics[width=0.49\textwidth,page=1]{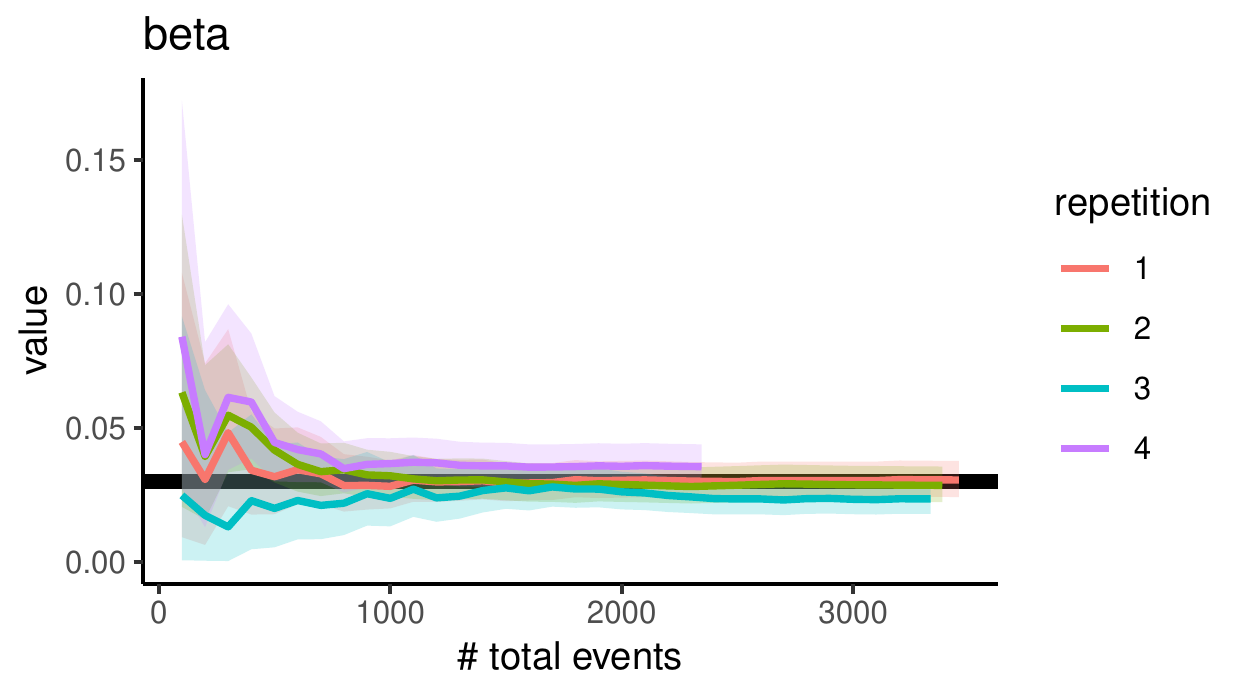}
    \includegraphics[width=0.49\textwidth,page=2]{figures/decoup_Bayes_1.pdf}
    \includegraphics[width=0.49\textwidth,page=3]{figures/decoup_Bayes_1.pdf}
    \includegraphics[width=0.49\textwidth,page=4]{figures/decoup_Bayes_1.pdf}
    \includegraphics[width=0.49\textwidth,page=5]{figures/decoup_Bayes_1.pdf}
    \includegraphics[width=0.49\textwidth,page=6]{figures/decoup_Bayes_1.pdf}
    \includegraphics[width=0.49\textwidth,page=7]{figures/decoup_Bayes_1.pdf}
    \includegraphics[width=0.49\textwidth,page=8]{figures/decoup_Bayes_1.pdf}
    \caption{Posterior sample means versus number of total events, estimated using datasets generated by the decoupled process model. True parameter values are marked by \textbf{bold dark} horizontal lines, along with 95\% credible bands. Results are presented for 4 different complete datasets.}
    \label{fig:Bayes-multi-decoup}
\end{figure}

\paragraph{Scalability of DARCI and the data-augmented inference scheme}
\added{In the main text, most simulations are conducted on a population of size $N=100$, in order to mimic the population size of the real data, but experiments have been carried out on larger networks (for example, with $N=200$ and $N=500$). Here we include some inference results for $N=500$ (see Figure~\ref{fig:post-hist-N500}). }

\begin{figure}[H]
    \centering
    \includegraphics[width=\textwidth]{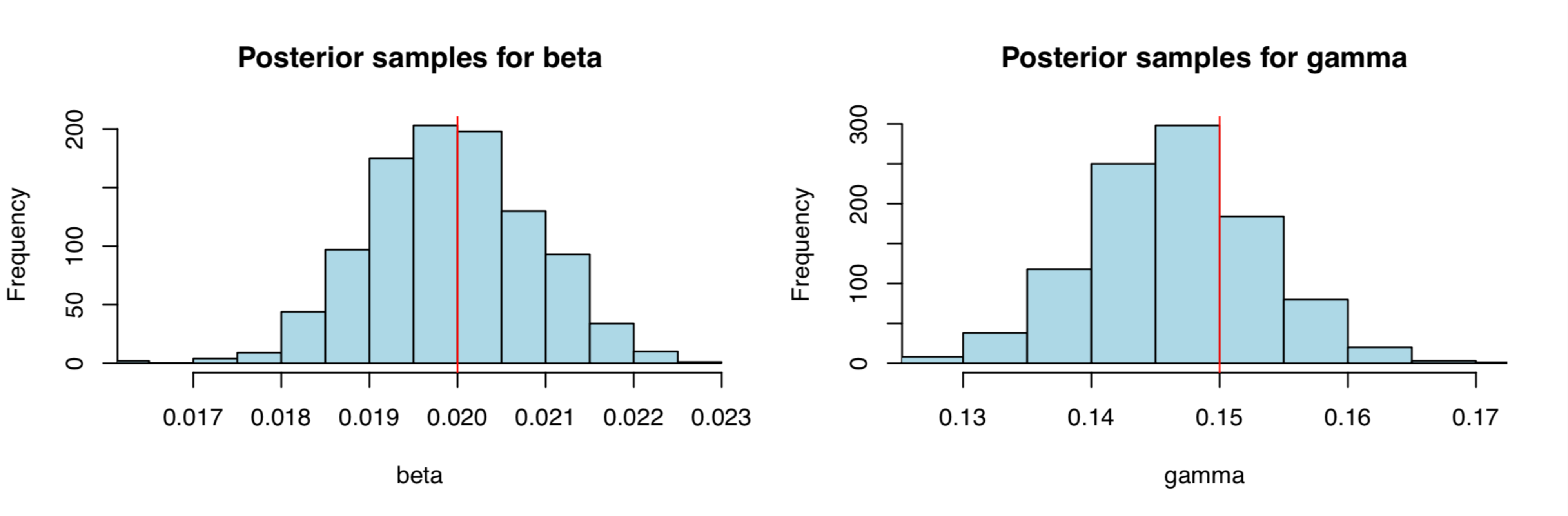}
    \caption{\added{Inference performance with population size $N=500$. Here we show histograms of posterior samples for parameter $\beta$ (infection rate, left panel) and $\gamma$ (recovery rate, right panel); ground truth is represented by red vertical lines. The proposed inference scheme can certainly handle a relatively large population with a lot of events, and can recover parameter values accurately.}}
    \label{fig:post-hist-N500}
\end{figure}

\added{It seems that the inference scheme is effective in recovering parameters and is capable of handling a large population and many events. Our finding is that computing time scales linearly with the total number of events observed. Figure~\ref{fig:run_time} shows the computing time on a single processor for each iteration with population size $N=500$, using a naive implementation in \texttt{R}. Since DARCI is parallelizable across time intervals, a more efficient implementation can further reduce computing time, and the algorithm bottleneck will be the maximum number of recovery events that occur in a time interval; moreover, the computation in one iteration only involves vector operations and random number sampling, which can be significantly sped up in other programming languages if necessary.  }

\begin{figure}[H]
    \centering
    \includegraphics[width=0.8\textwidth]{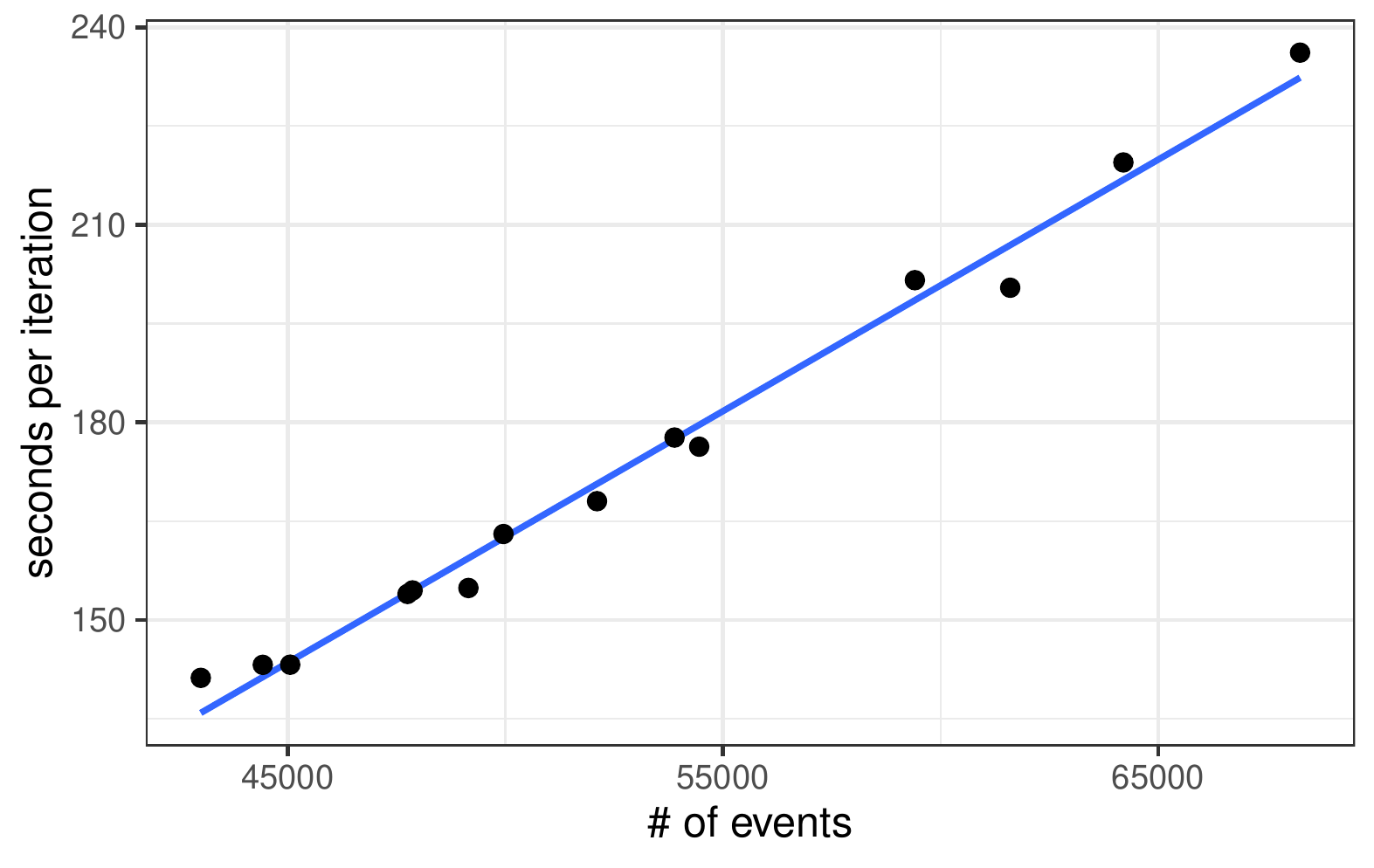}
    \caption{\added{Run time per iteration (in seconds) on \textbf{one processor} versus number of observed events in simulations, for population size $N=500$. Computing time scales linearly with the number of events, and the algorithm can handle data at least at the scale of $10^{4}$ events. }}
    \label{fig:run_time}
\end{figure}

\section{Real Data Experiments}
\label{app: real-data}
\ifthenelse{\boolean{tocut}}{

\subsection{Data Pre-processing}
All infection events and weekly health statuses of all $N=103$ individuals are extracted from the weekly surveys. In every survey, study participants were asked if they ever felt ill at all in the past week, if they ever experienced certain symptoms, and, if there were symptoms, when the approximate illness onset time was. We take an ``infection'' as a positive ILI (influenza-like illness) case, which, following the protocol in \cite{aiello2016design}, is defined as a cough plus at least one of the following symptoms: fever or feverishness, chills, or body aches. We further examine each ILI case and only accept one as a positive infection if the individual also indicated that they ``felt ill'' in the past week, thus eliminating a small number of reoccurring ILI cases for the same participants \footnote{One particular individual had positive ILI cases and felt ill in week 2, 3, and 5, but not in week 4. We therefore treat his/her illness as an extended one, starting in week 2 and lasting till week 5.}. Moreover, since an individual may start exhibiting symptoms at most 3 days \emph{after} getting infected and becoming infectious, for each infection event, we set the ``real'' infection time as the reported onset time minus a random ``delay time'' uniformly sampled between $0$ and $3$ days.

Social link activation and termination events are obtained from the iEpi Bluetooth contact records. Each time two study devices were paired, the iEpi application recorded the unique identifiers of the devices, a timestamp, and a received signal strength indicator (RSSI). Since Bluetooth detection can be activated whenever two devices are within a few meters of each other while the two users may not actually be in contact, we only keep those Bluetooth records with relatively strong signals (high values of RSSIs) \footnote{The RSSIs range from -109 to 6, and we set the threshold as -90, so only those records with RSSIs larger than -90 are kept.}. If two consecutive Bluetooth records for one pair of devices are no more than 7.5 minutes apart in time \footnote{We choose 7.5 minutes as a threshold instead of 5 minutes to accommodate potential lapses in Bluetooth detection.}, then the two records are considered to belong to one single continuous contact; a social link between two individuals is activated at the time of the first Bluetooth detection record in a series of consecutive records that belong to a single contact, and the link is terminated at a random time point between 1 and 6 minutes after the last Bluetooth detection of a continuous contact. 

The resulting processed data contain 24 infection events in total, with 14 before the spring break week and 10 after, as well as 45,760 social link activation and termination events. The weekly disease status (healthy or ill) of every participant can be acquired from the weekly surveys, so we know, for example, if an individual recovered sometime after day $7$ and before day $14$, but the exact times of all recoveries are unknown.

\subsection{Maximum Likelihood Estimation}
}{}
Instead of assuming the knowledge of which infection cases are internal and which are external, we directly estimate all the parameters based on the likelihood function in (\ref{eq:SIR-lik-external-1}), solving (\ref{eq:beta-partial-2}) and (\ref{eq:kappa-partial}) for the MLEs of $\beta$ and $\xi$. 

However, the real data are incomplete, with the exact times of all the recoveries unobserved. We resolve this issue using a naive imputation method---for each recovery, an event time is randomly sampled from a uniform distribution between the time of infection and the earliest time point the individual no longer felt ill (in response to the weekly surveys). Such imputation, of course, is subject to a considerable level of uncertainty, so we randomly generate 10 differently imputed datasets %for each observation period
, obtain the MLEs from every dataset, and then report the averages over the 10 runs (see Table~\ref{tab:realdata-res-MLE}).

We can see that the MLEs acquired in this manner generally agree with the Bayesian estimates in Section~\ref{sec:realdata-res}.

% \begin{table}[ht]
%     %\centering
%     \begin{tabular}{l||c|c|c}
%     \toprule
%     Data used &  $\beta$ & $\xi$ & $\gamma$ \\
%     \hline\hline
%     Joint &  $0.0676\pm0.0092$ & $0.00320\pm 1.11\times 10^{-6}$ & $0.236\pm0.012$\\
%     Period 1 &  $0.0774\pm0.0245$ & $0.00311\pm 3.39\times 10^{-6}$ & $0.220\pm0.018$\\
%     Period 2 & $0.0424\pm0.0065$ & $0.00328\pm7.34\times 10^{-7}$ & $0.259\pm 0.023$ \\
%     %\bottomrule
%     \end{tabular}
%     \begin{tabular}{l||c|c|c|c}
%     \toprule
%     Data used & $\alpha_{SS}$ & $\omega_{SS}$ & $\alpha_{SI}$ & $\omega_{SI}$ \\
%     \hline\hline
%     Joint & $0.0530\pm0.0001$ & $42.15\pm 0.105$ & $0.0704\pm0.0028$ & $52.21\pm 3.83$ \\
%     Period 1 & $0.0707\pm0.0005$ & $47.07\pm 0.31$ & $0.0840\pm0.0362$ & $58.30\pm 10.94$ \\
%     Period 2 & $0.0340\pm 0.0002$ & $34.31\pm0.16$ & $0.0387\pm0.0076$ & $32.83\pm5.08$ \\
%     \bottomrule
%     \end{tabular}
%     \caption{MLEs for select parameters using imputed data with all recovery times randomly sampled. The table presents average estimates as well as the standard deviations of estimates over 10 different, randomly imputed data. Inference is first carried out using data on both observation periods jointly (``Joint''), and then separately on the observation period before the spring break (``Period 1'') and that after (``Period 2'').
%     Results generally agree with those acquired using the proposed Bayesian data augmentation inference method.}
%     \label{tab:realdata-res-MLE}
% \end{table}

\begin{table}[ht]
    \centering
    \caption{MLEs for model parameters using imputed data with all recovery times randomly sampled. The table presents average estimates as well as the standard deviations of estimates over 10 different, randomly imputed datasets. 
    Results generally agree with those acquired using the proposed Bayesian data augmentation inference method.}
    \begin{tabular}{lrr}
    \toprule
    Parameter & Avg. estimate & Std. deviation \\
    \hline %\hline
    $\beta$ (internal infection) &  $0.0676$ & $0.0092$ \\
    $\xi$ (external infection) & $0.00320$ & $1.11\times 10^{-6}$ \\
    $\gamma$ (recovery) & $0.236$ & $0.012$ \\
    $\alpha_{SS}$ ($S$-$S$ link activation) & $0.0530$ & $0.0001$\\
    $\omega_{SS}$ ($S$-$S$ link termination) & $42.15$ & $0.105$\\
    $\alpha_{SI}$ ($S$-$I$ link activation) & $0.0704$ & $0.0028$ \\
    $\omega_{SI}$ ($S$-$I$ link termination) & $52.21$ & $3.83$\\
    \bottomrule
    \end{tabular}
    \label{tab:realdata-res-MLE}
\end{table}
%\end{appendices}

% comment out to put everything back
% \let\cleardoublepage\clearpage
% \processdelayedfloats 

\end{document}